\documentclass{article}

\usepackage[vmargin=1in,hmargin=1in]{geometry}
\usepackage{amsmath,amssymb,amsfonts,graphicx,amsthm,nicefrac,mathtools,bm}
\usepackage{hyperref}
\usepackage[numbers]{natbib}
\usepackage[english]{babel}
\usepackage{times}
\usepackage[T1]{fontenc}
\usepackage{bbm,mdframed}
\usepackage[dvipsnames]{xcolor}
\usepackage{xspace}
\usepackage{rotating,multirow}
\usepackage{graphics,tikz}
\usepackage{grffile}
\usepackage{subcaption}
\usepackage{epstopdf} 
\usepackage{enumitem} 
\usepackage[ruled]{algorithm}
\graphicspath{{./fwerplots/}}
\usepackage[toc,title,page]{appendix}
\usepackage[export]{adjustbox}
\usepackage{CJK}
\usepackage{times}
\usepackage{colortbl}

\newtheorem{theorem}{Theorem}

\newtheorem{lemma}{Lemma}

\newtheorem{proposition}{Proposition}

\theoremstyle{definition}
\newtheorem{definition}{Definition}
\theoremstyle{remark}
\newtheorem{remark}{Remark}

\makeatletter

\newtheorem*{rep@theorem}{\rep@title}
\newcommand{\newreptheorem}[2]
{\newenvironment{rep#1}[1]
	{\def\rep@title{#2 \ref{##1}} \begin{rep@theorem}}%
		{\end{rep@theorem}}}
\makeatother
\newreptheorem{theorem}{Theorem}
\newreptheorem{lemma}{Lemma}
\newreptheorem{corollary}{Corollary}
\newreptheorem{proposition}{Proposition}

\newcommand{\figref}[1]{Figure~\ref{fig:#1}}

\newcommand{\secref}[1]{Section~\ref{sec:#1}}
\newcommand{\appref}[1]{Appendix~\ref{app:#1}}

\newcommand{\defref}[1]{Definition~\ref{def:#1}}

\newcommand{\lemref}[1]{Lemma~\ref{lem:#1}}

\newcommand{\propref}[1]{Proposition~\ref{prop:#1}}

\newcommand{\thmref}[1]{Theorem~\ref{thm:#1}}

\newcommand{\eqnref}[1]{equation~\eqref{#1}}

\newcommand{\one}[1]{{\mathbbm{1}}_{{#1}}}

\newcommand{\PP}[1]{\textnormal{Pr}\!\left\{{#1}\right\}} 
\newcommand{\EE}[1]{\mathbb{E}\left[{#1}\right]} 

\newcommand{\EEN}[1]{\mathbb{E}_{N}\left[{#1}\right]}
\newcommand{\EEst}[2]{\mathbb{E}\left[{#1}\ \middle| \ {#2}\right]} 
\newcommand{\PPst}[2]{\text{Pr}\!\left\{{#1}\ \middle| \ {#2}\right\}} 

\def\R{\mathbb{R}}

\newcommand{\ignore}[1]{}

\newcommand{\diff}[1]{\text{d}{#1}} 

\newcommand{\nulls}{\mathcal{H}_0}
\newcommand{\PFER}{\textnormal{PFER}}
\newcommand{\FWER}{\textnormal{FWER}}
\newcommand{\nullsc}{\mathcal{H}_0^{c}}

\newcommand{\seq}[1]{\{{#1}\}_{i=1}^{\infty}}

\usepackage{lipsum}

\newcommand\blfootnote[1]{%
	\begingroup
	\renewcommand\thefootnote{}\footnote{#1}%
	\addtocounter{footnote}{-1}%
	\endgroup
}

\usepackage{fancyhdr}
\newcommand{\thedate}{\today}
\newcommand{\theauthor}{}
\newcommand{\thetitle}{Online control of the familywise error rate \blfootnote{An R package called \href{http://www.texample.net/tikz/resources/}{onlineFDR}\cite{online} began with algorithms for online FDR control, but now contains all aspects of online multiple testing, including online FWER control. It is being developed by David Robertson and the authors of this paper (among others), and has already incorporated the new algorithms proposed here. The code to reproduce all figures in this paper are accessible at \href{https://github.com/jinjint/onlineFWER}{https://github.com/jinjint/onlineFWER}.}}

\date{\thedate}
\author{\theauthor}
\title{\thetitle}

\fancypagestyle{plain}{%
	\fancyhf{}%
	\lhead{\thepage}
}

\usepackage[mmddyy]{datetime}

\newcommand{\fdr}{\textnormal{FDR}}
\newcommand{\fwer}{\textnormal{FWER}}
\newcommand{\pfer}{\textnormal{PFER}}

\def\S{\mathcal{S}}
\def\C{\mathcal{C}}
\def\R{\mathcal{R}}
\def\F{\mathcal{F}}

\def\N{\mathbb{N}}
\def\V{\mathcal{V}}

\makeatletter
\newcommand{\dotfrac}[2]{
	\mathchoice
	{\ooalign{$\genfrac{}{}{0pt}{0}{#1}{#2}$\cr\leavevmode\cleaders\hb@xt@ .22em{\hss $\displaystyle\cdot$\hss}\hfill\kern\z@\cr}}
	{\ooalign{$\genfrac{}{}{0pt}{1}{#1}{#2}$\cr\leavevmode\cleaders\hb@xt@ .22em{\hss $\textstyle\cdot$\hss}\hfill\kern\z@\cr}}
	{\ooalign{$\genfrac{}{}{0pt}{2}{#1}{#2}$\cr\leavevmode\cleaders\hb@xt@ .22em{\hss $\scriptstyle\cdot$\hss}\hfill\kern\z@\cr}}
	{\ooalign{$\genfrac{}{}{0pt}{3}{#1}{#2}$\cr\leavevmode\cleaders\hb@xt@ .22em{\hss $\scriptscriptstyle\cdot$\hss}\hfill\kern\z@\cr}}
}
\makeatother

\usepackage{algorithm,algorithmic}

\newcommand{\defn}{\ensuremath{:\, =}}

\tikzstyle{none} = [rectangle, rounded corners, minimum width=0.5cm, minimum height=0.5cm,text centered, draw=black, fill=red!20]
\tikzstyle{noneoff} = [rectangle, rounded corners, minimum width=0.5cm, minimum height=0.5cm,text centered, draw=black, fill=red!8]

\tikzstyle{ada} = [rectangle, rounded corners, minimum width=0.5cm, minimum height=0.5cm, text centered, draw=black, fill=orange!20]
\tikzstyle{dis} = [rectangle, rounded corners, minimum width=0.5cm, minimum height=0.5cm, text centered, draw=black, fill=blue!20]
\tikzstyle{adaoff} = [rectangle, rounded corners, minimum width=0.5cm, minimum height=0.5cm, text centered, draw=black, fill=orange!7]

\tikzstyle{adadis} = [rectangle, rounded corners, minimum width=0.5cm, minimum height=0.5cm, text centered, draw=black, fill={rgb:red,10;green,20;yellow,54}]
\tikzstyle{adadisoff} = [rectangle, rounded corners, minimum width=0.5cm, minimum height=0.5cm, text centered, draw=black, fill=green!5]

\tikzstyle{arrow} = [thick,->,>=stealth, text width=3cm]
\tikzstyle{dotarrow} = [dashed, text width=3cm]

\tikzstyle{title} = [rectangle, rounded corners, minimum width=1cm, minimum height=1cm, text centered, draw=white, fill=none]

\begin{document}
	\author{
	Jinjin Tian, Aaditya Ramdas\\
	Department of Statistics and Data Science\\
	Carnegie Mellon University\\
	\texttt{\{jinjint,aramdas\}@stat.cmu.edu}
}
\maketitle

\begin{abstract}
{Biological research often involves testing a growing number of null hypotheses as new data is accumulated over time. 
We study the problem of online control of the familywise error rate (FWER), that is testing an apriori unbounded sequence of hypotheses ($p$-values) one by one over time without knowing the future, such that with high probability there are no false discoveries in the entire sequence. This paper unifies algorithmic concepts developed for offline (single batch) FWER control and online false discovery rate  control to develop novel online FWER control methods.  Though many offline FWER methods (e.g. Bonferroni, fallback procedures and Sidak's method) can trivially be extended to the online setting, our main contribution is the design of new, powerful, adaptive online algorithms that control the FWER 
when the $p$-values are independent or locally dependent in time. Our experiments demonstrate substantial gains in power, that are also formally proved in a Gaussian sequence model.}{Multiple testing, FWER control, online setting.}
\end{abstract}

   \section{Introduction}\label{sec:intro}
    Modern genomics studies typically require large scale multiple hypotheses testing, which are sometimes conducted in an online manner (meaning one or few hypotheses at a time), not as a big large batch of hypotheses tested all at once. Thus the family of tested hypotheses is  continually growing over time, due to the accumulation of data (both types and amounts). For example, one international scientific project, the International Mouse Phenotyping Consortium (IMPC), aims to create and characterize the phenotype of 20,000 knockout mouse strains; this was launched in September 2011, and was projected to take 10 years. Available datasets and the resulting family of hypotheses constantly grow as new knockouts are studied, while discovery-inspired downstream analyses are conducted along the way.
    Thus we are faced with a nonstandard multiple hypothesis testing problem, one in which the nature or number of hypotheses (or $p$-values) are not known in advance but become known one at a time. This exemplifies the challenge of online hypothesis testing, where at each step the scientist must decide whether or not to reject the current null hypothesis without knowing the future hypotheses or their outcomes (rejected or not). 
   
	Formally speaking, online multiple testing refers to the setting in which a potentially infinite stream of hypotheses $H_1,H_2,\dots$ (respectively $p$-values $P_1,P_2,\dots$) is tested  one by one over time. It is important to distinguish the aforementioned setup from sequential hypothesis testing, where a single hypothesis is repeatedly tested as new data to test that hypothesis are collected. Our setup zooms out one level, and it is the hypotheses that appear one at a time, each one summarized by a $p$-value.
	
	At each step $t \in \mathbb{N}$, one must decide whether to reject the current null hypothesis $H_t$ or not, without knowing the outcomes of all the future tests. Typically, we reject the null hypothesis when $P_t$ is smaller than some threshold $\alpha_t$.  Let $\R(T)$ represent the set of rejected null hypotheses by time $T$, and $\nulls$ be the unknown set of true null hypotheses; then, $\V(T) \equiv \R(T) \cap \nulls$ is the set of incorrectly rejected null hypotheses, also known as false discoveries. Denoting $V(T)=|\V(T)|$, some error metrics are the false discovery rate (FDR), familywise error rate (FWER) and power which are defined as
	\begin{equation}
		\fdr(T) \equiv \EE{\frac{V(T)}{|\R(T)|\vee 1}},\ \ \ \fwer(T) \equiv \PP{V(T) \geq 1}, \ \ \ \textnormal{power}(T) \equiv \EE{\frac{|\nulls^{c}\cap \R(T)|}{|\nulls^{c}|}}.\label{fdr}
	\end{equation}
	
 It is arguably of interest to control the FDR or FWER at each time below a prespecified level $\alpha\in(0,1)$, while maximizing power. Such a goal is  also essential with respect to  providing a reliable resource for downstream analysis along the way. Notice that FWER$(T)$ is monotonically nondecreasing in T, therefore controlling FWER$(T)$ at each $T \in \N$ is equivalent to controlling $\lim_{T\to\infty}\textnormal{FWER}(T)$. Therefore we drop the index $T$, and only discuss techniques controlling  $\lim_{T\to\infty}\textnormal{FWER}(T)$.
	
	There is a large variety of procedures for (A) offline (single batch) FDR control \citep{BH95, storey2002direct, zhao2018multiple}, (B) online FDR control \citep{foster2008alpha, aharoni2014generalized, javanmard2017on, ramdas2017on, ramdas2018saffron, tian2019addis}, and (C) offline FWER control \citep{holm1979simple, hochberg1988sharper, wiens2005fallback}. However, the online FWER problem is underexplored; for example, in their seminal online FDR paper,  \citet{foster2008alpha} only mention in passing that for online FWER control, one can simply use an online Bonferroni method (that they term Alpha-Spending). We first point out that offline FWER procedures can be easily extended to have online counterparts (like Sidak's and fallback methods). However, our main contribution is to derive new ``adaptive'' procedures for online FWER control that are demonstrably more powerful under independence or local (in time) dependence assumptions.
	
	In experiments, we find that our online extensions of classical offline FWER control methods like Sidak \citep{vsidak1967rectangular} and fallback \citep{wiens2005fallback,burman2009recycling} have a fairly negligible power improvement over Alpha-Spending when non-nulls are interspersed randomly over time. Also, as presented in \figref{intro}, these algorithms are sub-optimal when they encounter a non-vanishing proportion of non-nulls, or a significant proportion of \emph{conservative} null\footnote{A common assumption on the marginal distribution of a null $p$-value $P$ is that for all $x \in [0,1]$, $\PP{P\leq x}\leq x$ holds true.  Ideally, a null $p$-value is exactly uniformly distributed, but in practice, null $p$-values are often \emph{conservative}, meaning that $\PP{P\leq x} \ll x$.}.
	In contrast, our adaptive discarding (ADDIS) algorithms are more powerful online FWER control methods that adapt to both an unknown fraction of non-nulls and unknown conservativeness of nulls (see \secref{addis} and \figref{history}).  As shown in \figref{intro}, ADDIS-Spending is significantly more powerful than Alpha-Spending, online Sidak, online fallback, etc. We are able to theoretically justify this power superiority in an idealized Gaussian model (\secref{power}), while also confirm it in the real IMPC dataset we mentioned earlier (\secref{real}). Although these powerful new adaptive methods require independence to have valid FWER control, they can be easily generalized to handle ``local dependence'', an arguably more realistic dependence structure in practice (see \secref{local}). 

\begin{figure}
		\centering
		\includegraphics[width=0.34\textwidth, valign=c]{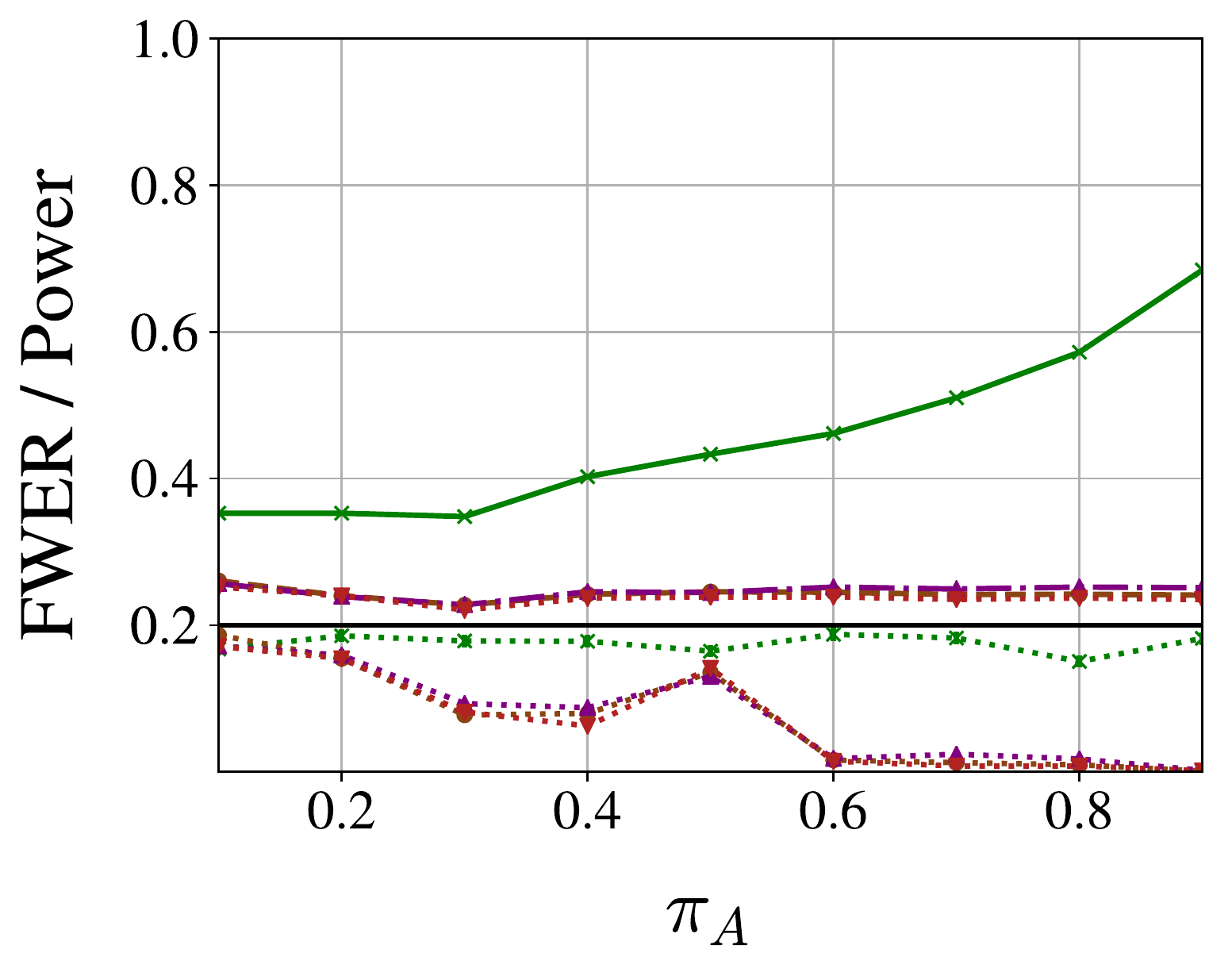}
		\includegraphics[width=0.34\textwidth, valign=c]{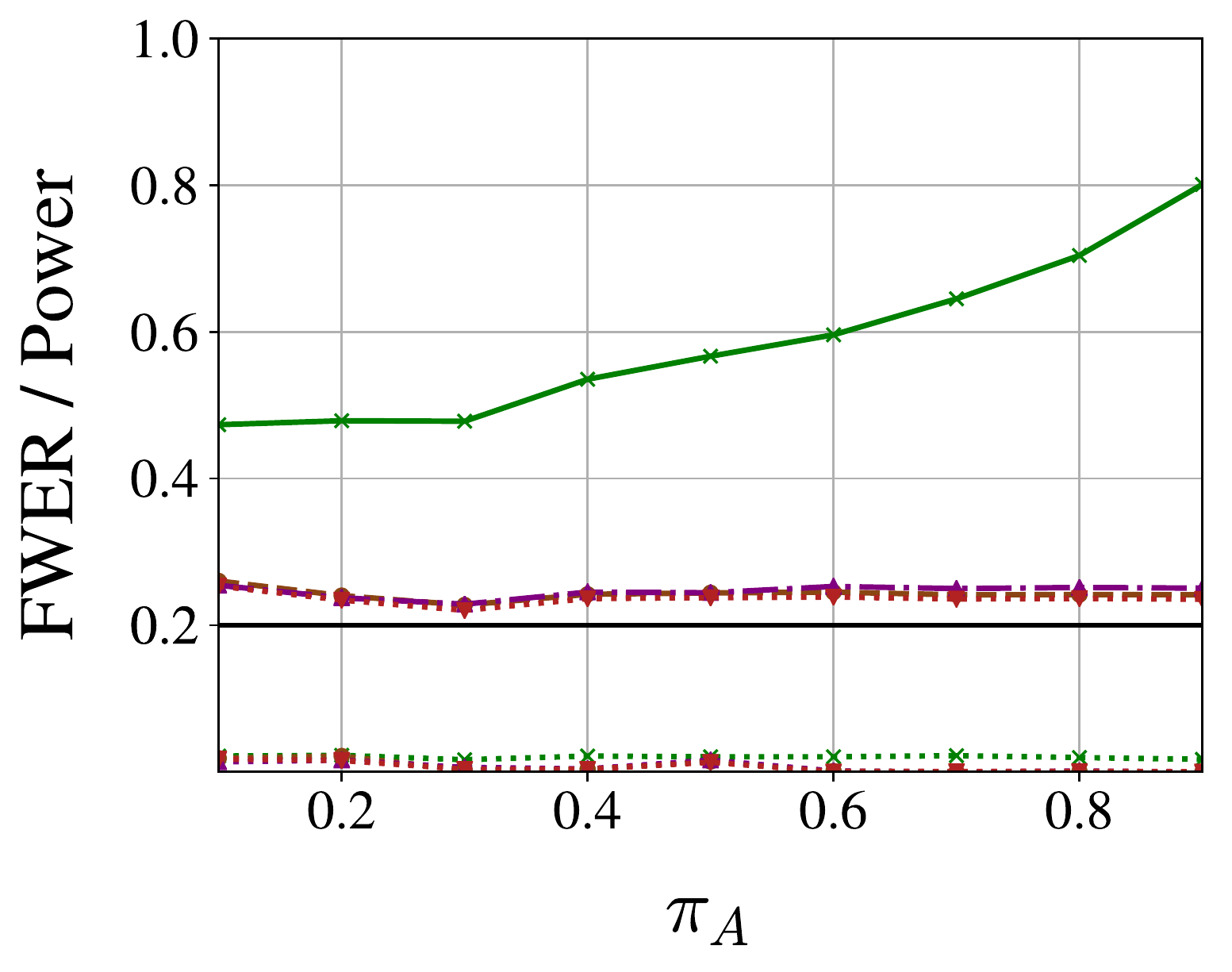}
		\includegraphics[width=0.17\textwidth, valign=c]{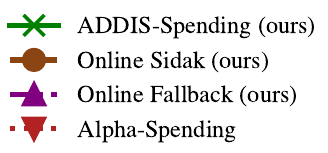} 
		\caption{Statistical power and FWER versus fraction of non-null hypotheses $\pi_A$ for Alpha-Spending (online Bonferroni) and our new algorithms (ADDIS-Spending, Online Fallback and Online Sidak) at target FWER level $\alpha= 0.2$ (solid black line). The curves above the horizontal line at $0.2$ display the achieved power of each methods versus $\pi_A$, while the lines below 0.2 display the achieved FWER of each methods versus $\pi_A$. The experimental setting is described in \secref{simu}: all observations are Gaussians, we set the alternative mean $\mu_A = 4$ for both figures, but we set the null mean $\mu_N = 0$ for the left figure and $\mu_N = -1$ for the right figure (hence the right nulls are conservative, the left nulls are not). These figures show that (a) all considered methods control the FWER at level 0.2, (b) Online Sidak and Online Fallback provide negligible improvement compared to the naive Alpha-Spending, (c) ADDIS-Spending is much more powerful than all other algorithms in both settings.}\label{fig:intro}
	\end{figure}   

	\begin{minipage}{0.9\columnwidth}
		\begin{center}
			\begin{tikzpicture}[node distance=1.1cm]
			\node (offline) [title] {Offline};    
			\node (bonf) [noneoff, right of=offline, xshift = 1cm] {\scriptsize{Bonferroni}};
			\node (adabonf) [adaoff, right of=bonf, xshift = 5cm, align = center] {\scriptsize{Adaptive} \scriptsize{Bonferroni} \\ \citep{schweder1982plots} \\ \citep{ hochberg1990more}};
			
			\node (online) [title, below of=offline, yshift = -2cm] {Online};
			\node (spend) [none, right of=online, xshift = 2cm, align = center] {\scriptsize{ Alpha-Spending}\\ \citep{foster2008alpha}};
			\node (ada-spend) [ada, right of=spend,  xshift = 5cm, yshift = 0.8cm, align = center] {\scriptsize{Adaptive-Spending}\\ \scriptsize{ (\secref{adaptive})}};
			\node (dis-spend) [dis, below of=ada-spend, yshift = -0.8cm, align = center] {\scriptsize{Discard-Spending}\\ \scriptsize{ (\secref{discarding})}};  
			\node (addis-spend) [adadis, right of=spend, xshift = 10cm, align = center] {\scriptsize{ADDIS-Spending} \\ \scriptsize{ (\secref{addis})}};
			
			
			\draw [arrow] (bonf) --node[anchor=south, align = center]{\scriptsize{adaptivity}}(adabonf);
			\draw [arrow] (spend) --node[anchor=south, align = center]{\scriptsize{adaptivity}}(ada-spend);
			\draw [arrow] (spend) --node[anchor=north, align = center]{\scriptsize{discarding}}(dis-spend);
			\draw [arrow] (ada-spend) --node[anchor=south, align = center]{}(addis-spend);
			\draw [arrow] (dis-spend) --node[anchor=south, align = center]{}(addis-spend);
			\draw [dotarrow] (bonf) --node[anchor=west]{\scriptsize{online analog}}(spend);
			\draw [dotarrow] (adabonf) --node[anchor=west]{\scriptsize{online analog}}(ada-spend);
			
	\end{tikzpicture}
	
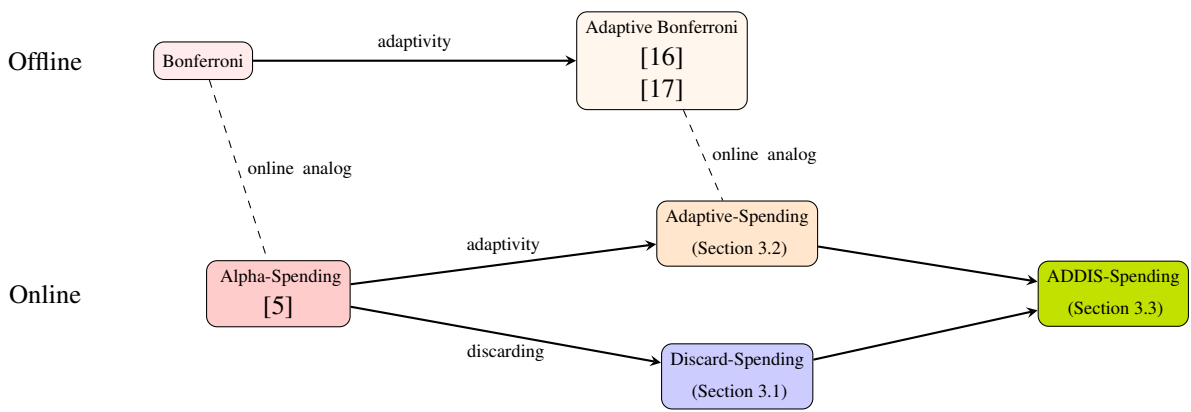
\captionof{figure}{Historical context for comparing adaptive offline and online methods for FWER control.}\label{fig:history}
	\end{center}
	\end{minipage}
	
	\paragraph{Paper outline}
	In \secref{current}, we extend some offline FWER methods to the online setting, and prove their FWER control. Then we derive ADDIS-Spending in \secref{ADDIS-Spending}, while we present its generalizations to handle local dependence in \secref{local}. \secref{power} contains theoretical results about the statistical power of the online FWER control procedures (both known and newly derived), particularly \secref{spending} derives optimal choices for hyperparameters of naive Alpha-Spending within an idealized Gaussian setup, and \secref{more} provides some theoretical justification of ADDIS-Spending being more powerful than naive Alpha-Spending.  Numerical studies are included in \secref{numerical} to empirically demonstrate the improved power of our new methods using both synthetic and real data. In the end, we also present several extensions of our main contribution, together with interesting followup directions for future work in \secref{conclude}.

	\section{Online extensions of classical offline FWER control methods}\label{sec:current}
	As we mentioned before, one may regard Alpha-Spending as an online generalization of Bonferroni correction in the offline setting. Specifically, for FWER level $\alpha$, given an infinite nonnegative sequence $\seq{\gamma_i}$ that sums to one, the Alpha-Spending tests individual hypothesis $H_i$ at level
	\begin{equation}
	    \alpha_i \defn \alpha \gamma_i. \label{spend}
	\end{equation}
	
	In this section, we seek online variants of other offline FWER methods such as Holm's \citep{holm1979simple} and Hochberg's methods \citep{hochberg1988sharper}; the fallback procedure \citep{wiens2005fallback} and alpha-recycling \citep{burman2009recycling},  and the correction proposed by \citet{vsidak1967rectangular}. Gladly, most of the methods (like Sidak, fallback and its generalization) can be trivially extended to the online setting, while maintaining strong FWER control, while some of them (like Holm and Hochberg) which depends on ordering are hard to incorporate with the online setting. In the following, we present two general classes of algorithms which uniformly improves Bonferonni, together with numerical analysis of their performance.  Specifically, those algorithms are Online Sidak (requires independence), an online analog of offline Sidak correction  \citep{vsidak1967rectangular}; and Online Fallback (works under arbitrary independence), an online analog of the offline generalized fallback  \citep{bretz2009graphical}. We expect more online procedures to be developed from offline methods in the future.
	
	\subsection{\textbf{Online Sidak}}
	One well known improvement of Bonferroni is the Sidak correction proposed by \citet{vsidak1967rectangular} which requires independence. Given $m$ different null hypotheses and FWER level $\alpha$, the Sidak method uses the testing level $1-(1-\alpha)^{\frac{1}{m}}$ for each hypothesis. Analogously, Online Sidak chooses an infinite nonnegative sequence $\seq{\gamma_i}$ which sums to one and tests each hypothesis $H_i$ at level 
	\begin{equation}\label{onsidak}
		\alpha_i \defn 1 - (1-\alpha)^{\gamma_i}.
	\end{equation}
	Just as Sidak is only slightly less stringent than Bonferroni, Online Sidak also improves very little of Alpha-Spending, as shown in \figref{denser}. In addition, Online Sidak also requires independence for valid FWER control, as stated below.
	\begin{proposition}\label{prop:sidak}
		Online Sidak controls FWER in a strong sense if the null $p$-values are independent with each other, and is at least as powerful as the corresponding Alpha-Spending procedure.
	\end{proposition}

	\begin{proof}
		The probability of no false discovery among all infinite decisions is 
		\[
		\PP{V=0} = \EE{\prod_{i \in \nulls} \one{P_i > \alpha_i} } \geq \prod_{i \in \nulls} (1-\alpha_i) = \prod_{i \in \nulls} (1-\alpha)^{\gamma_i} = (1-\alpha)^{\sum_{i \in \nulls} \gamma_i} \geq 1-\alpha.
		\]
		Further, since $\alpha_i > \gamma_i \alpha$, Online Sidak is at least as powerful as the corresponding Alpha-Spending procedure.
	\end{proof}

	\subsection{Online fallback procedures}
	The fallback procedure \citep{wiens2005fallback} and its graphical generalization \citep{bretz2009graphical} for offline FWER control can be easily extended to the online setting, improving the power of Alpha-Spending by ``recycling'' the significance level of previous rejections. Below, we briefly present the online variants of the generalized fallback \citep{bretz2009graphical}, specifically named Online Fallback.
	\paragraph{\textbf{Online Fallback}} For FWER level $\alpha$, Online Fallback chooses an infinite sequence $\seq{\gamma_i}$ that sums to one, and tests each $H_i$ at level 
	\begin{equation}\label{recycling}
		\alpha_i := \gamma_i \alpha + \sum_{k=1}^{i-1}w_{k,i} R_{k} \alpha_{k},
	\end{equation}
    where $R_i = \one{P_i \leq \alpha_i}$, and $\{w_{ki}\}_{i = k+1}^{\infty}$ being some infinite sequence that is nonnegative and sums to one for all $k \in \N$. Specifically, we call $\{w_{ki}\}_{i = k+1}^{\infty}$ the transfer weights for the $k$-th hypothesis. \\
	
	\noindent
	Note that when we choose $\{w_{ki}\}_{i = k+1}^{\infty}$ with $w_{ki} = \one{k=i-1}$, the individual testing level of Online Fallback reduces to 
    \begin{equation}\label{eqsaving}
	 \alpha_i := \alpha \gamma_i  + R_{i-1} \alpha_{i-1}
	\end{equation}
	for all $i\geq 2$, and $\alpha_1 = \alpha\gamma_1$, which happens to be the offline fallback procedure \citep{wiens2005fallback} applying on the infinite hypotheses sequence. We specifically refer to this procedure as \textbf{Online Fallback-1}. 
	
	The following Proposition states the FWER control of Online Fallback and is proved in \appref{recycling}. In fact, Online Fallback can be treated as an instantiation of a more general sequential rejection principle proposed by \citet{goeman2010sequential}, and \propref{recycling} follows trivially from their results. However, here we present a more tangible proof for \propref{recycling} from a separate direction, since it allows us to get more insights about the connection between the offline methods and their online variants. 
	\begin{proposition}\label{prop:recycling}
		Online Fallback controls FWER for arbitrarily dependent $p$-values, and is at least as powerful as the corresponding Alpha-Spending procedure.
	\end{proposition}
	\noindent
    Though \propref{recycling} states the power superiority of Online Fallback over Alpha-Spending, one may have the intuition that the power improvement of Online Fallback over Alpha-Spending would be little when encounters randomly arrived signals, since much of the recycled significance level are inevitably wasted on the nulls. \figref{denser} demonstrate this intuition, where the left figure describes the setting of randomly arrived signals and plot against varied signal frequency $f$; while the right figure describes the setting of clustered signals at the beginning of the sequence and plot against varied cluster range $r$. As shown in the left figure of \figref{denser}, where the non-nulls are evenly distributed, Online Fallback has basically the same power with Alpha-Spending. 
    Meanwhile, in the extreme case when the beginning of the sequence consists of only non-nulls (e.g. when $r = 1$ in the right figure), Online Fallback (especially Online Fallback-1) are much more powerful than Alpha-Spending, while this advantage diminishes fast as more non-nulls plug in the signal cluster (i.e. as $r$ increases). 

	\begin{figure}
		\begin{subfigure}{\textwidth}
     	\includegraphics[width=0.4\textwidth]{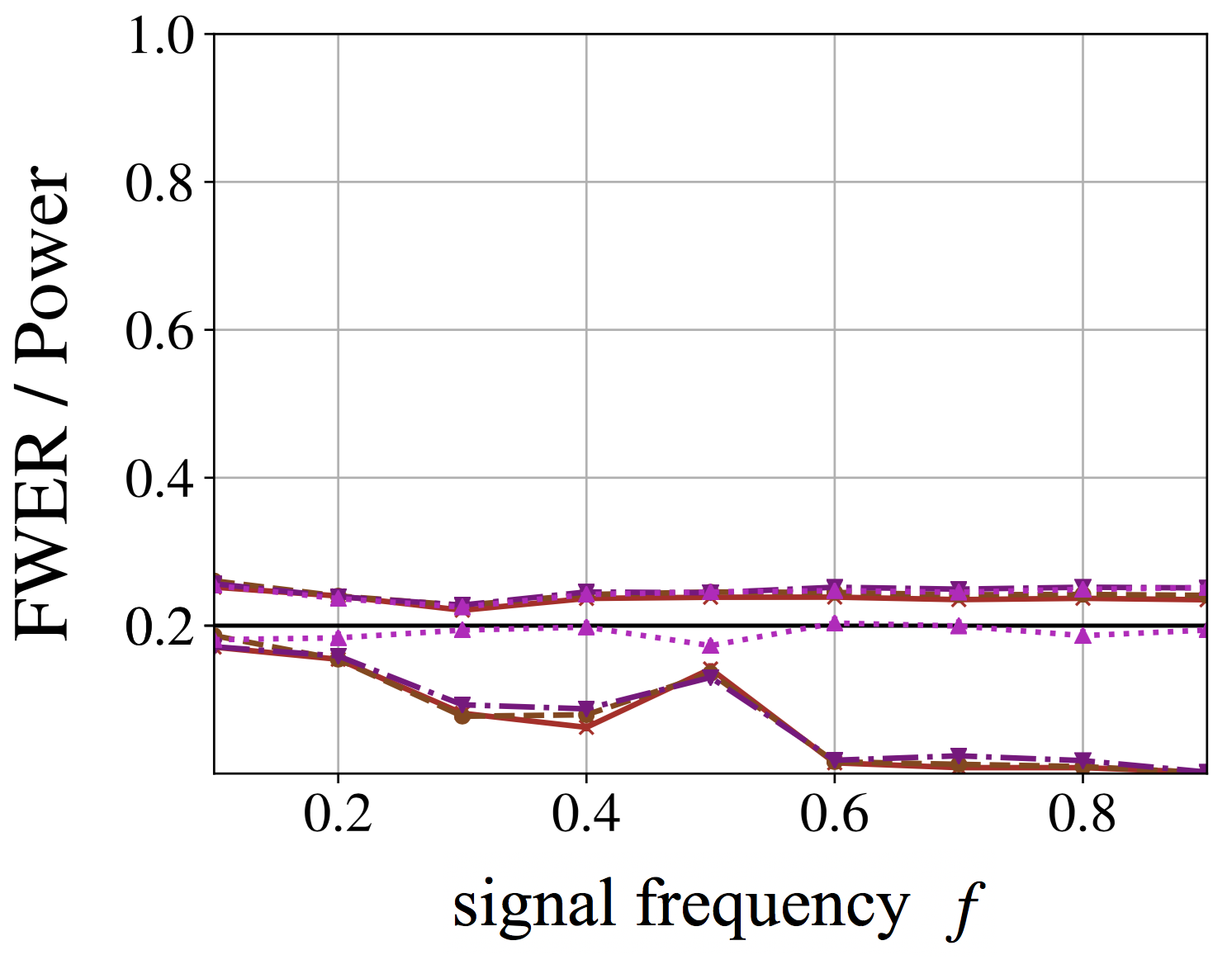}
    	\includegraphics[width=0.4\textwidth]{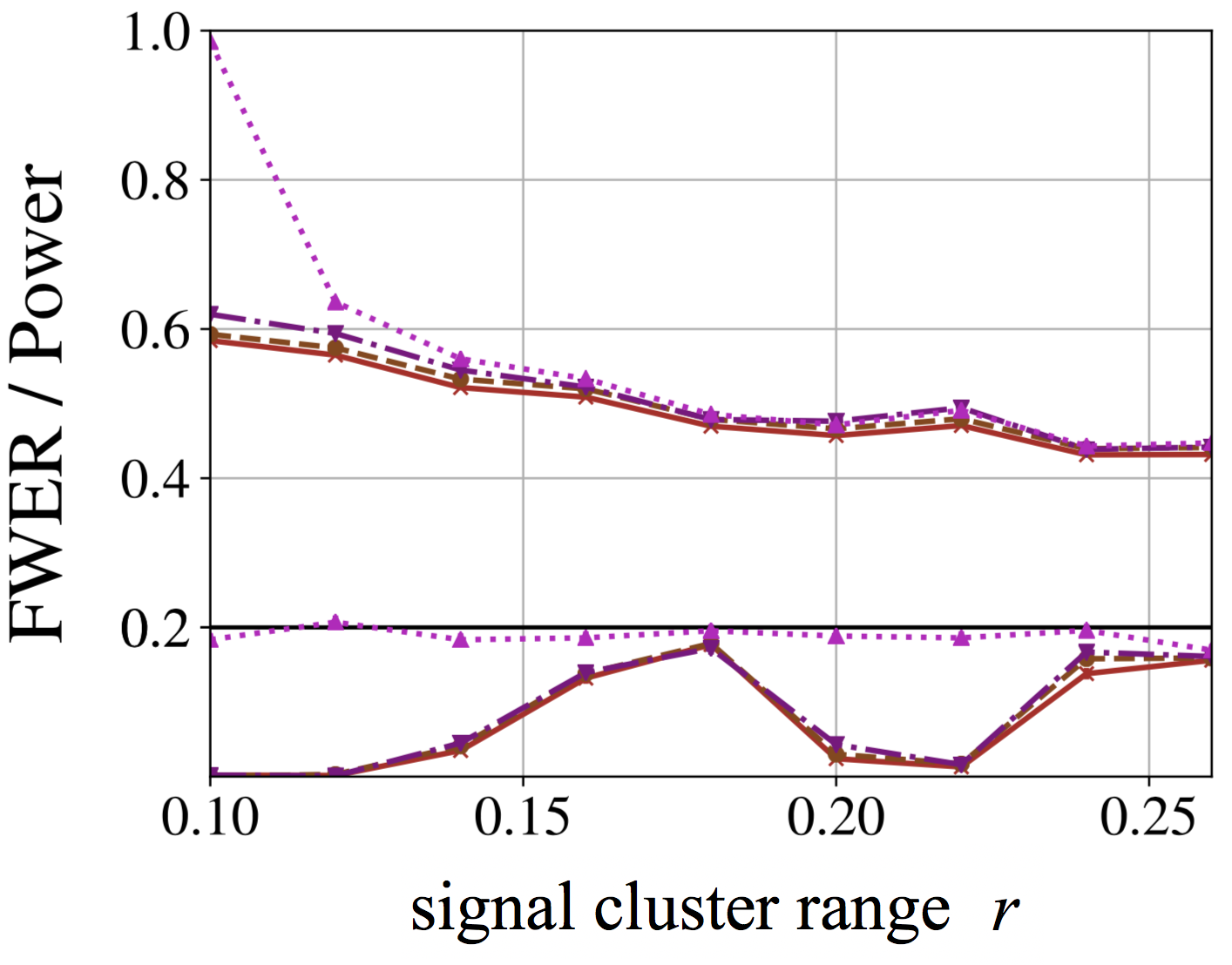}
    	\includegraphics[width=0.18\textwidth]{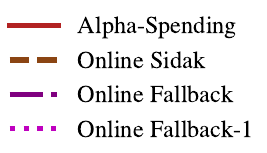}
		\end{subfigure}
		\caption{Statistical power and FWER for Alpha-Spending, Online Sidak, Online Fallback-1 and Online Fallback at target FWER level $\alpha= 0.2$ (solid black line). The curves above 0.2 line display the power of each methods, while the lines below 0.2 display the FWER of each methods.  The specific experimental setting follows those in  \secref{simu}, with few alterations. The main experimental setting is described in \secref{simu}, while we allow $\pi_A$ to be different for each hypotheses $H_i$, and denote it as $\pi_{Ai}$ instead. We set $\mu_N \equiv 0, \mu_A \equiv 4$ for all the figures; and we set $\pi_{A_i} = f $ for $i$ less than $\lfloor T r\rfloor$ and $\pi_{A_i} = 0$ otherwise. Particularly, we set $f \in \{0.1, 0.2, \dots, 0.9\}$ for the left figure, while $f \equiv 0.1$ for the right figure; and we set $r = 1$ for the first figure, 
	    and $r \in \{0.1, 0.12, \dots, 0.26\}$ for the right figure. The underlying $\seq{\gamma_i}$ is $\seq{6/\pi^2 i^2}$ for all the methods, and particularly we set  $\{w_{ki}\}_{i=k+1}^{\infty}$ with $w_{ki} = \gamma_{i-k}$ for all $k \in \N$ in Online Fallback. These figures display that (a) all the considered methods do control FWER at level 0.2; (b) the new methods (Online Sidak, Online Fallback-1 and Online Fallback) have almost the same power as Alpha-Spending given evenly distributed non-nulls; (c) Online Fallback-1 and Online Fallback have noticeable power improvement over Alpha-Spending when non-nulls cluster in the beginning of the testing sequence, with greater improvement if the non-nulls are more compact in the beginning of the testing sequence. }\label{fig:denser}
	\end{figure}
     
     The above new online FWER control methods are guaranteed to be uniformly more powerful than Alpha-Spending, though the improvements are usually minor, except for extreme cases when the non-nulls are clustered at the beginning of the testing sequence. This motivates the development of a new class of algorithms in the following section.
	
	\section{Adaptive discarding (ADDIS) algorithms}\label{sec:ADDIS-Spending}
	In \secref{current}, we refined Alpha-Spending by refining its offline variants. Since we find that the resulting methods rarely improve power much over Alpha-Spending, we refine Alpha-Spending from another angle in this section, directly addressing the looseness in the proof of its FWER control. We end up with a series of new adaptive algorithms that are much more powerful than Alpha-Spending, though at the cost of requiring independence. We ease the requirement in \secref{local} to allow the methods to handle a local dependence structure, which is a more practical and reliable dependence structure as discussed in \cite{zrnic2019asynchronous}.
	
	First, we explain the looseness in the proof of FWER control of Alpha-Spending. Recall the Alpha-Spending procedure in \eqref{spend}. Its FWER is controlled because an even stronger error metric, the per-family wise error rate (PFER) is also controlled at $\alpha$:
	\begin{equation}\label{pfspending}
		\PFER \defn \EE{V} = \EE{\sum_{i\in\nulls} \one{P_i \leq \alpha_i} } = \sum_{i \in \nulls} \EE{\one{P_i \leq \alpha_i}} \overset{(\mathbbm{A})}{\leq} \sum_{i \in \nulls} \alpha_i \overset{(\mathbbm{B})}{\leq} \sum_{i \in \N} \alpha_i = \alpha.
	\end{equation} 
    \noindent
    (One can verify that FWER $\leq$ PFER by Markov's inequality.)
	The above proof indicates that the PFER of Alpha-Spending may be far less than $\alpha$ (and hence the power would be much lower than necessary) if either of the two inequalities in \eqref{pfspending} are loose. Specifically, the first inequality $(\mathbbm{A})$ would be loose if the null $p$-values are very conservative, and the second inequality $(\mathbbm{B})$ would be loose if there are a large number of non-nulls. These facts expose two weaknesses of Alpha-Spending: it will become suboptimal when encounters conservative nulls or a constant fraction (like 20\%) of non-nulls, both are arguably quite likely in real applications. These weaknesses are also problems for the methods derived in \secref{current}, since they have almost the same power with Alpha-Spending in those simulations.
	
	By addressing these two weaknesses, we develop a more powerful family of methods called ADDIS-Spending, which are adaptive discarding algorithms that not only benefit from adaptivity to the fraction of nulls, but also gain by exploiting conservative nulls (if they exist). Instead of directly presenting the ADDIS-Spending algorithm, we first introduce the idea of discarding and adaptivity in \secref{discarding} and \secref{adaptive} respectively, each addressing one of the looseness mentioned above, and in \secref{addis} we combine the ideas together to develop our ADDIS-Spending algorithm, which addresses both loosenesses. There is a price to pay for these improvements. Alpha-Spending works even when the $p$-values are arbitrarily dependent, but the idea of discarding and adaptivity essentially requires independence between $p$-values. We ease this requirement later in \secref{local}, by generalizing our ADDIS-Spending algorithm to handle a local dependence structure. 
	
	Before we proceed, it is useful to set up some notation. Recall that $P_j$ is the $p$-value resulted from testing hypothesis $H_j$. Given infinite sequences $\{\tau_t\}_{t=1}^{\infty}$, $\{\lambda_t\}_{t=1}^{\infty}$ and $\{\alpha_t\}_{t=1}^{\infty}$, where every element is in $(0,1)$, define the indicators
	\begin{equation}\label{indicators}
		S_j = \one{P_j \leq \tau_j}, \quad C_j = \one{ P_j \leq \lambda_j }, \quad R_j = \one{ P_j \leq \alpha_j},    
	\end{equation}
	which respectively indicate whether $H_j$ is selected for testing (used for adapting to conservative nulls), whether $H_j$ is a candidate for rejection (used for adapting to fraction of nulls), and whether $H_j$ is rejected. Accordingly define $R_{1:t}= \{R_1,\dots,R_{t}\}$, $C_{1:t}= \{C_1,\dots,C_{t}\}$ and $S_{1:t} =\{S_1,\dots,S_{t}\}$. Similarly, let $\R = \{i\in \N: R_i =1\}$, $\C = \{i\in \N: C_i =1\}$, $\S = \{i\in \N: S_i =1\}$. In what follows, we say  $\alpha_t$, $\lambda_t$ and $\tau_t$ are ``predictable'' to mean that they are measurable with respect to some filtration $\F^{t-1}$, meaning that they are mappings from $\F^{t-1}$ to $(0,1)$. The form of $\F^{t-1}$ may change across algorithms and it is denoted $\sigma(\cdot)$, that is a sigma field of certain random variables.
	
	\subsection{Discarding conservative nulls}\label{sec:discarding}
	Here, we develop a method named Discard-Spending to address the first looseness $(\mathbbm{A})$ in \eqnref{pfspending} due to overlooking the conservativeness of null $p$-values. Specifically, we consider the null $p$-values to be \emph{uniformly conservative} as defined below, which include uniform nulls (the ideal case) and the majority of conservative nulls (realistic case).
	
	\begin{definition}\label{def:conserve} (Uniformly conservative)
    A null $p$-value $P$ is said to be uniformly conservative if for all $x,\tau \in [0,1]$, 
	\begin{equation}\label{conserve-def}
		\PPst{P \leq x\tau}{P \leq \tau} \leq x.
	\end{equation} 
	\noindent
	\end{definition}
	The above condition is easy to interpret; for example when $\tau=0.6$ and $x=0.5$, it means the following---if we know that $P$ is at most $0.6$, then it is more likely between $[0.3,0.6]$ than in $[0,0.3]$.
	As an obvious first example, uniform null $p$-values are uniformly conservative. As for more general examples, note that the definition  \eqnref{conserve-def} is mathematically equivalent to requiring that for the CDF $F$ of the $p$-value $P$, we have
	\begin{equation}\label{conserve-def-equal}
		F(\tau x) \leq x F(\tau) \quad \text{ for all } x,\tau \in [0,1].
	\end{equation}
	\noindent
	 Hence, null $p$-values with convex CDF are uniformly conservative. 
	Moreover, noting that monotonically nondecreasing density implies convex CDF for null $p$-values, \citet{zhao2018multiple} presented the following \propref{expo}.
	
	\begin{proposition}\label{prop:expo}
	\citep{zhao2018multiple} For a one-dimensional exponential family with true parameter $\theta \leq \theta_0$, the $p$-value resulting from the uniformly most powerful (UMP) test of $H_0: \theta \leq \theta_0$ versus $H_1: \theta > \theta_0$ is uniformly conservative.
	\end{proposition}
	Particularly, the corresponding nulls in \propref{expo} are exactly uniform at the boundary of the null set ($\theta=\theta_0$), while conservative at the interior. 
	Since the true underlying state of nature is rarely \emph{exactly} at the boundary, it is common in practice to encounter uniformly conservative nulls that are indeed conservative. One simple example following \propref{expo} is the one-sided test of Gaussian mean, where we test the null hypothesis $H_0 : \mu  \leq 0$ against the alternative $H_1: \mu > 0$ with observation $Z \sim N(\mu,1)$.  From \propref{expo}, the $p$-value resulting from UMP test, that is $P = \Phi(-Z)$ is uniformly conservative, and is conservative if $\mu < 0$, while the conservativeness increases as $\mu$ decreases, where $\Phi$ is the standard Gaussian CDF.
	
	\begin{remark}
	As pointed out by \cite{ellis2017gaining}, the \emph{uniformly conservative} condition is also known as the uniform conditional stochastic order (UCSO) defined by \cite{whitt1980uniform} relative to uniform distribution $U(0,1)$. \cite{zhao2018multiple} also mentioned that their argument with regard one-dimensional exponential family that we stated above can also be viewed as a special case of Theorem 1.1 proposed by  \cite{whitt1980uniform}. We refer the readers to \cite{whitt1980uniform} for more details.
	\end{remark} 

	There have been some works addressing uniformly conservative $p$-values in the offline setting \citep{zhao2018multiple, ellis2017gaining}, which both boil down to one simple idea: \emph{discard} (do not test) large $p$-values, and test the others after some rescaling.  In particular, \cite{zhao2018multiple} used it in the global null test setting while \cite{ellis2017gaining} used it with regard  FWER/FDR control. It has recently been utilized for more powerful online FDR control by \cite{tian2019addis}. Here we extend this \emph{discarding} idea to the online setting for FWER control, stated specifically as the following Discard-Spending algorithm.
	
	\paragraph{\textbf{Discard-Spending}} Recalling the definitions in and right after \eqnref{indicators}, we call any online FWER algorithm as a Discard-Spending algorithm if it updates the $\alpha_i$ in a way such that $\seq{\alpha_i}$ satisfy the following conditions: (1) $\alpha_i$ is predictable, where the filtration $\F^{i-1} = \sigma(R_{1:i-1}, S_{1:i-1})$; (2) $\tau_i$ is predictable, and $\alpha_i < \tau_i$ for all $i \in \N\ $ and \begin{equation}\label{eq:discarding}
	\sum_{i \in \S} \frac{\alpha_i}{\tau_i} \leq \alpha.\\
	\end{equation}
	
	It is not hard to show that Discard-Spending is equivalent to the following strategy: if $P_i > \tau_i$, then we do not test it (we \emph{discard} it), and if $P_i \leq \tau_i$, we test $P_i $ at some level $\alpha_i \in (0,1)$ satisfying \eqref{eq:discarding}. As a concrete example of Discard-Spending, choose a nonnegative sequence $\seq{\gamma_i}$ that sums to one, and test every $H_i$ at the adapted level 
	\begin{equation}\label{concretedis}
		\alpha_i := \alpha \tau_i \gamma_{t(i)}, \quad \text{where } t(i) = 1+\sum_{j < i} S_j.
	\end{equation}
	With simple algebra, one can see that this procedure satisfies $\sum_{i\in \S }\frac{\alpha_i}{\tau_i} \leq \alpha$,  $\alpha_i < \tau_i$, and that $\alpha_i$ is $\F^{i-1}$-measurable for all $i\in \N$. In the following, we claim the FWER control of general Discard-Spending algorithms. 
	
	\begin{proposition}\label{prop:discard}
		Discard-Spending controls PFER and thus FWER in a strong sense when the null $p$-values are uniformly conservative as defined in \eqnref{conserve-def}, while being independent of each other and of the non-nulls. 
	\end{proposition}
	\begin{proof}
	Recall that V is the number of false discoveries, which can be written as 
	$V = \sum_{i \in \nulls} R_i S_i$, using the fact that $\alpha_i < \alpha < \tau_i$ by construction. Therefore, 
	\begin{equation}
	    \EE{V} = \sum_{i \in \nulls} \EE{R_i S_i} = \sum_{i\in\nulls} \EE{\EEst{R_i S_i }{S_i,\F^{i-1}}} 
			= \sum_{i\in\nulls} \EE{\EEst{R_i}{S_i=1,\F^{i-1}}\PP{S_i=1}}
	\end{equation}
	where we used linearity of expectation twice. Since $\alpha_i, \lambda_i, \tau_i \in \F^{i-1}$, therefore each term inside the last sum equals
		\begin{align}
			\EE{\PPst{\frac{P_i}{\tau_i} \leq \frac{\alpha_i}{\tau_i}}{P_i \leq \tau_i,\F^{i-1}} \PP{S_i=1}} &\leq \EE{\frac{\alpha_i}{\tau_i} \PP{S_i=1}}
		\end{align}
		using the property of uniformly conservative nulls in \eqnref{conserve-def}. Then using the fact that $\alpha_i < \tau_i$ are measurable with regard $\F^{i -1}$, we have that
		\begin{align}
			\EE{\frac{\alpha_i}{\tau_i} \PP{S_i=1}} &= \EE{\frac{\alpha_i}{\tau_i} \EEst{\one{P_i \leq \tau_i}}{S_i = 1, \F^{i-1}} \PP{S_i=1}} \nonumber\\
			&= \EE{ \EEst{\frac{\alpha_i}{\tau_i}\one{P_i \leq \tau_i}}{S_i = 1, \F^{i-1}} \PP{S_i=1}} \nonumber\\
			& = \EE{\EEst{ \frac{\alpha_i}{\tau_i} \one{ P_i \leq \tau_i}}{S_i, \F^{i-1}}} = \EE{\frac{\alpha_i}{\tau_i} \one{ P_i \leq \tau_i}},
		\end{align}
		 where we again use the law of iterated expectation in the last step.
		Therefore
		\begin{equation}\label{pfdiscard}
		    \EE{V} \leq \sum_{i\in \nulls}\EE{\frac{\alpha_i}{\tau_i} \one{ P_i \leq \tau_i}} =\EE{\sum_{i\in \nulls \cap \S} \frac{\alpha_i}{\tau_i}}\stackrel{(i)}{\leq}  \EE{\sum_{i\in \S} \frac{\alpha_i}{\tau_i}} \leq\alpha
		\end{equation}
		by construction. Therefore, $\pfer \leq \alpha$ and thus $\fwer \leq \alpha$ as claimed.
	\end{proof}

	Discarding can lead to higher power when there are many conservative nulls. However, inequality (i) in \eqref{pfdiscard} will be really loose if we have $|\nulls \cap \S| \ll |\S|$, meaning that most of the contents of $\S$ are non-nulls, which is possible since the indices in $\S$ are those with small $p$-values. \figref{discarding} demonstrates both the strength and weakness of Discard-Spending, where we use $\alpha_i$ as described in the concrete example \eqref{concretedis} mentioned above, and set $\tau_i=0.5$ and $\gamma_i = \frac{6}{\pi^2 i^2}$ for all $i$.  
	
	\begin{figure}
		\centering
		\includegraphics[scale=0.33]{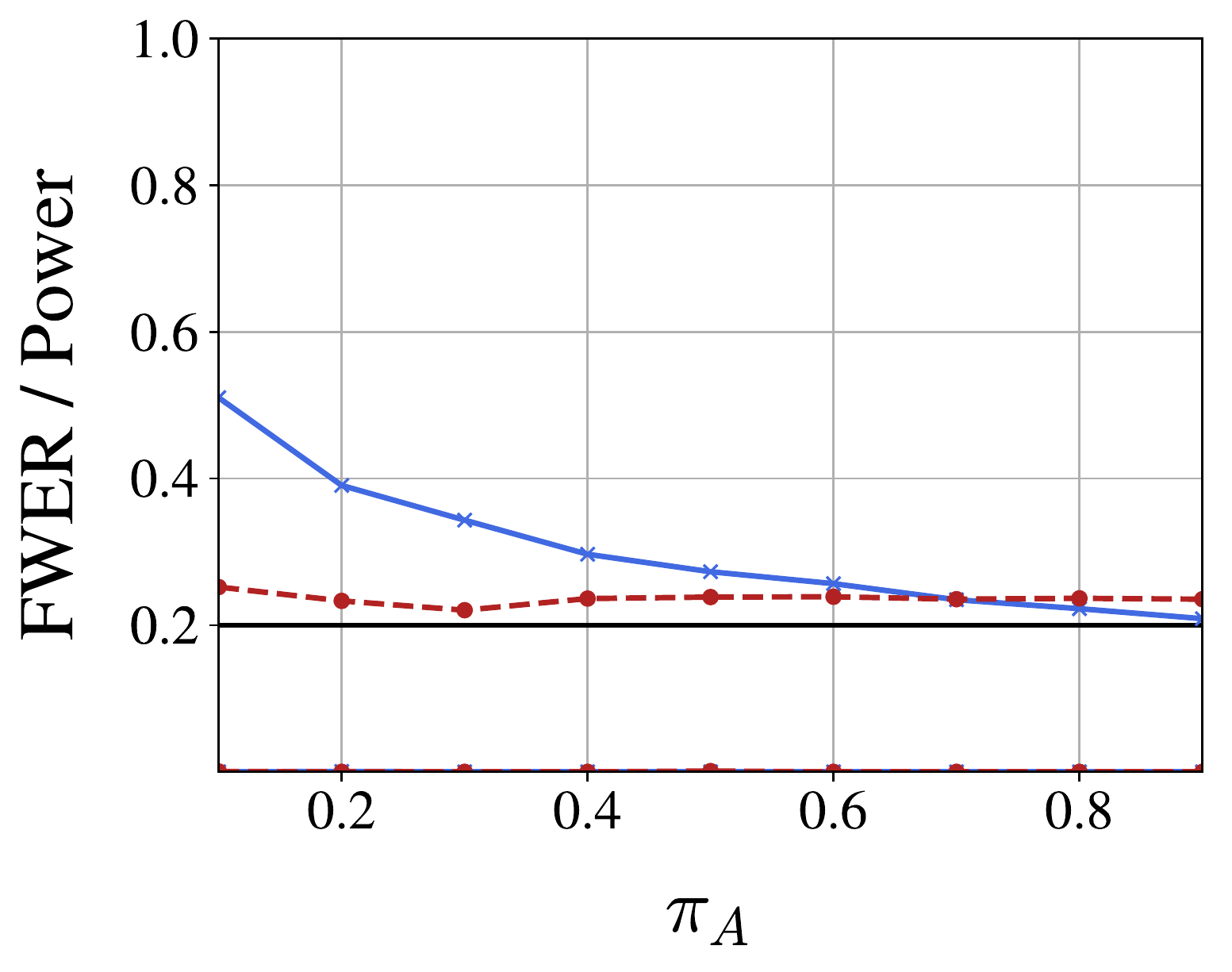}
		\includegraphics[scale=0.33]{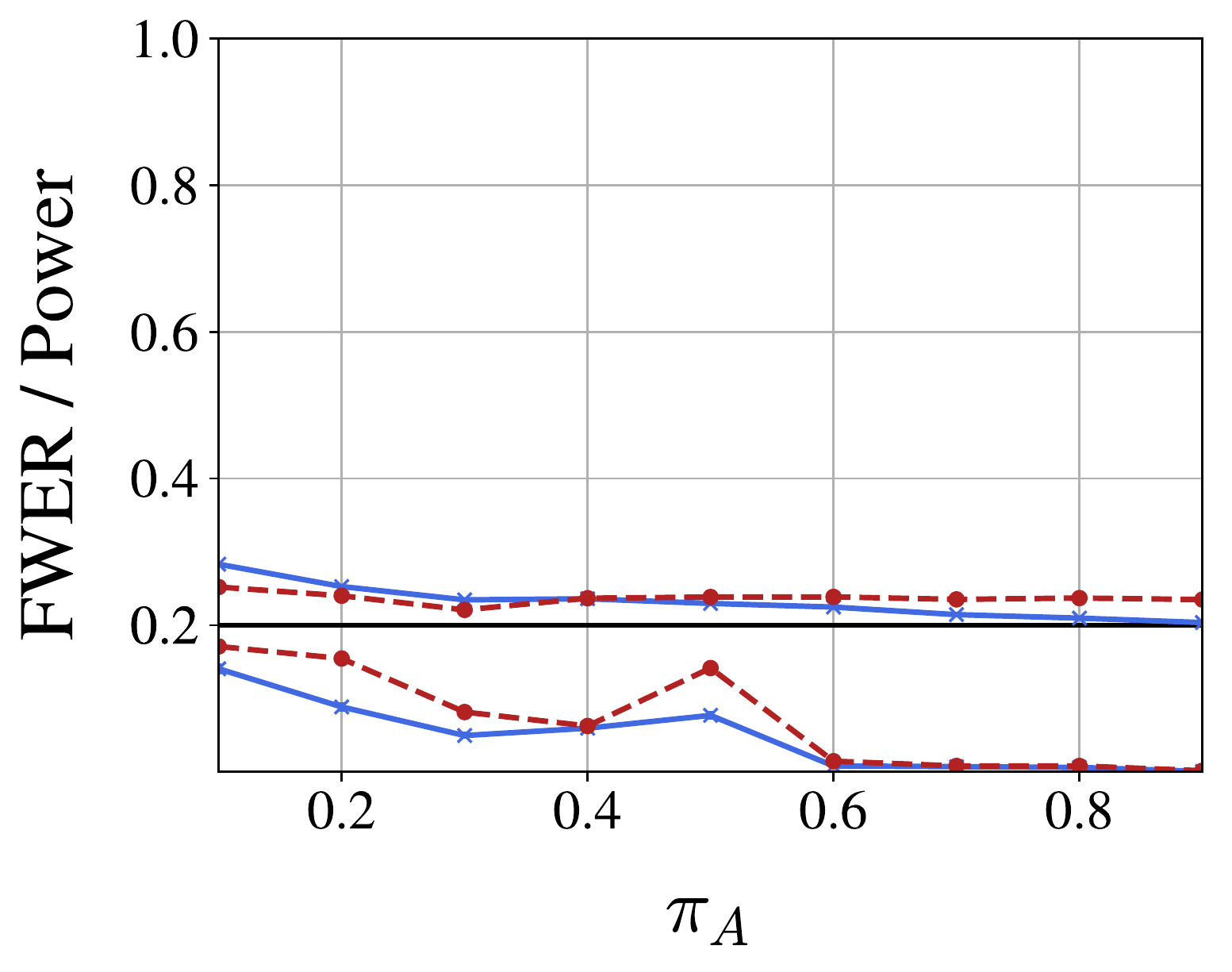}\\
		\includegraphics[scale=1.1]{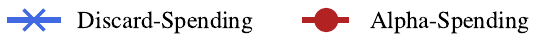}
		\caption{Statistical power and FWER versus fraction of non-null hypotheses $\pi_A$ for Discard-Spending and  Alpha-Spending at target FWER level $\alpha= 0.2$ (solid black line). The curves above 0.2 line display the power of each methods versus $\pi_A$, while the lines below 0.2 display the FWER of each methods versus $\pi_A$. The experimental setting is described in \secref{simu}: we set $\mu_A = 4$ for both figures, but $\mu_N = -2$ for the left figure and $\mu_N =0$ for the right figure (hence the left nulls are conservative, the right nulls are not). These figures show that: (a) Discard-Spending do control FWER  at level 0.2; (b) Discard-Spending is more powerful than naive Alpha-Spending when the nulls are conservative (as shown in the left figure); (c) Discard-Spending loses power when encounter high proportion of non-nulls (as shown in the right figure).}\label{fig:discarding}
	\end{figure}
	
	

	\subsection{Adaptivity to unknown proportion of nulls.}\label{sec:adaptive}
	Here, we develop a method to address the second looseness $(\mathbbm{B})$ in \eqnref{pfspending}, which is due to lack of adjustment for the proportion of true nulls. 
	Related ideas have been proposed during the development of offline FWER methods, which can be improved by incorporating an estimate of the true null proportion. This led to a series of adaptive methods like the adaptive Bonferroni  \citep{schweder1982plots, hochberg1990more}, and other generalizations like adaptive Sidak, adaptive Holm and adaptive Hochberg, which are rigorously proved to have FWER control by  \citet{finner2009controlling, guo2009note} and  \citet{sarkar2012adaptive}. Inspired by those efforts, we introduce the Adaptive-Spending procedure next, which could be regarded as an online variant of adaptive Bonferroni. 
	
	\paragraph{\textbf{Adaptive-Spending}} Recalling the definitions in and right after \eqnref{indicators}, we call any online FWER algorithm as Adaptive-Spending if it updates $\alpha_i$ in a way such that $\seq{\alpha_i}$ satisfy the the following conditions: 
	(i) $\alpha_i$ is predictable, where the filtration $\F^{i-1} = \sigma(R_{1:i-1}, C_{1:i-1})$; (ii) for a predictable sequence $\seq{\lambda_i}$, we have 
	\begin{equation}\label{eq:adaptivity}
	    \sum_{i \notin \C} \frac{\alpha_i}{1-\lambda_i} \leq \alpha.
	\end{equation}
	
	In other words, the testing process of Adaptive-Spending is that, whenever $P_i \leq \lambda_i$, we don't lose any error budget for testing it at level $\alpha_i$; but whenever $P_i > \lambda_i$, we lose $\alpha_i/(1-\lambda_i)$ from our error budget. As a concrete example of Adaptive-Spending, suppose $\seq{\gamma_i}$ is a nonnegative sequence that sums to one, then we test each $H_i$ at the predictable level 
	\begin{equation}\label{concretead}
		\alpha_i \defn \alpha(1-\lambda_i) \gamma_{t(i)}, \quad \text{where } t(i) = i - \sum_{j < i} C_j.
	\end{equation}
	\noindent
	Again with simple algebra, one may easily verify that procedure \eqref{concretead} satisfies all the conditions in the definition of Adaptive-Spending algorithm. 
	Next, we prove the FWER control of Adaptive-Spending. 
	\begin{proposition}\label{prop:thmadaptive}
		Adaptive-Spending controls PFER and thus FWER, when the null $p$-values are independent of each other and of the non-nulls.
	\end{proposition}	
	\begin{proof}
		We prove the PFER (and hence FWER) control of Adaptive-Spending using the law of iterated expectation and the fact that the null $p$-values are valid, i.e. $\PP{P_i\leq x} \leq x$ for any $x \in [0,1]$ and $i \in \nulls$. Recalling that $V$ is the number of false discoveries, specifically, we have
		\begin{align}\label{pfadaptive}
			\EE{V} &= \EE{ \sum_{i \in \nulls} \one{P_i \leq \alpha_i} } 
			= \sum_{i \in \nulls} \EE{\EEst{\one{P_i \leq \alpha_i}}{\F^{i-1}}}\stackrel{(i)}{\leq} \sum_{i \in \nulls} \EE{\alpha_i}\nonumber\\&
			\stackrel{(ii)}{\leq} \sum_{i \in \nulls} \EE{\alpha_i \EEst{\frac{\one{P_i > \lambda_i}}{1-\lambda_i}}{\F^{i-1}}}
			= \sum_{i \in \nulls} \EE{\alpha_i \frac{\one{P_i > \lambda_i}}{1-\lambda_i}}
			\leq \EE{\sum_{i \in \N} \alpha_i \frac{\one{P_i > \lambda_i}}{1-\lambda_i}} = \EE{\sum_{i \notin \C} \frac{\alpha_i}{1-\lambda_i}}
			\leq \alpha.
		\end{align}
		Therefore, $\pfer \leq \alpha$ and thus $\fwer \leq \alpha$ as claimed.
	\end{proof}	
	
	Adaptive procedures can improve power substantially if there is a non-negligible proportion of signals. However, their power can also suffer considerably if the null $p$-values are very conservative, since inequalities (i) and (ii) in \eqref{pfadaptive} would be extremely loose when $\PP{P_i \leq x} \ll x$ and $\PP{P_i > x} \gg 1-x$. These strengths and weaknesses of Alpha-Spending are demonstrated in \figref{adaptive}, where we use $\alpha_i$ as described in the concrete example \eqref{concretead} mentioned above, and set $\lambda_i= 0.5$ and $\gamma_i = \frac{6}{\pi^2 i^2}$ for all $i \in \N$. 
	
	\begin{figure}
		\centering
		\includegraphics[scale=0.33]{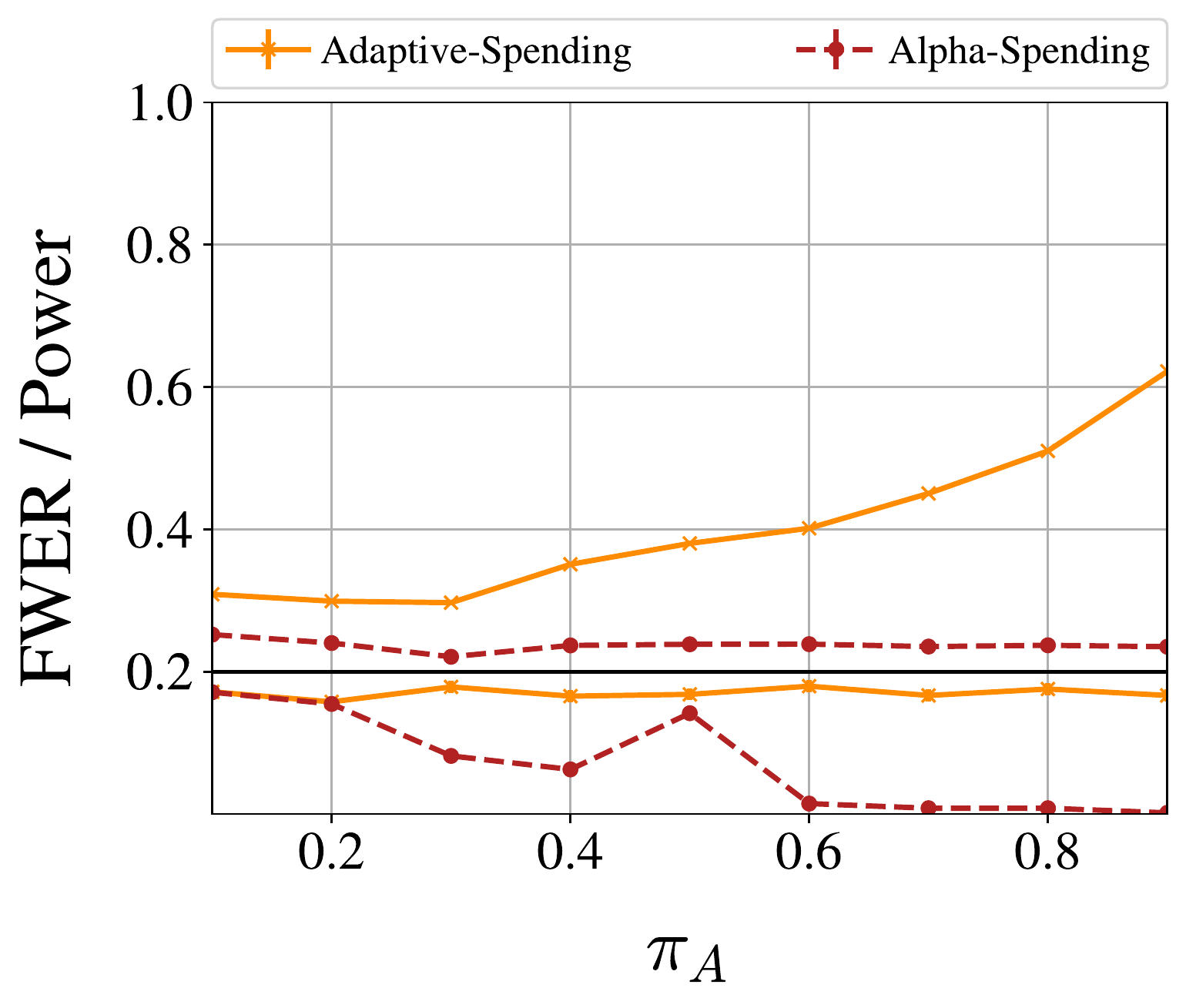}
		\includegraphics[scale=0.33]{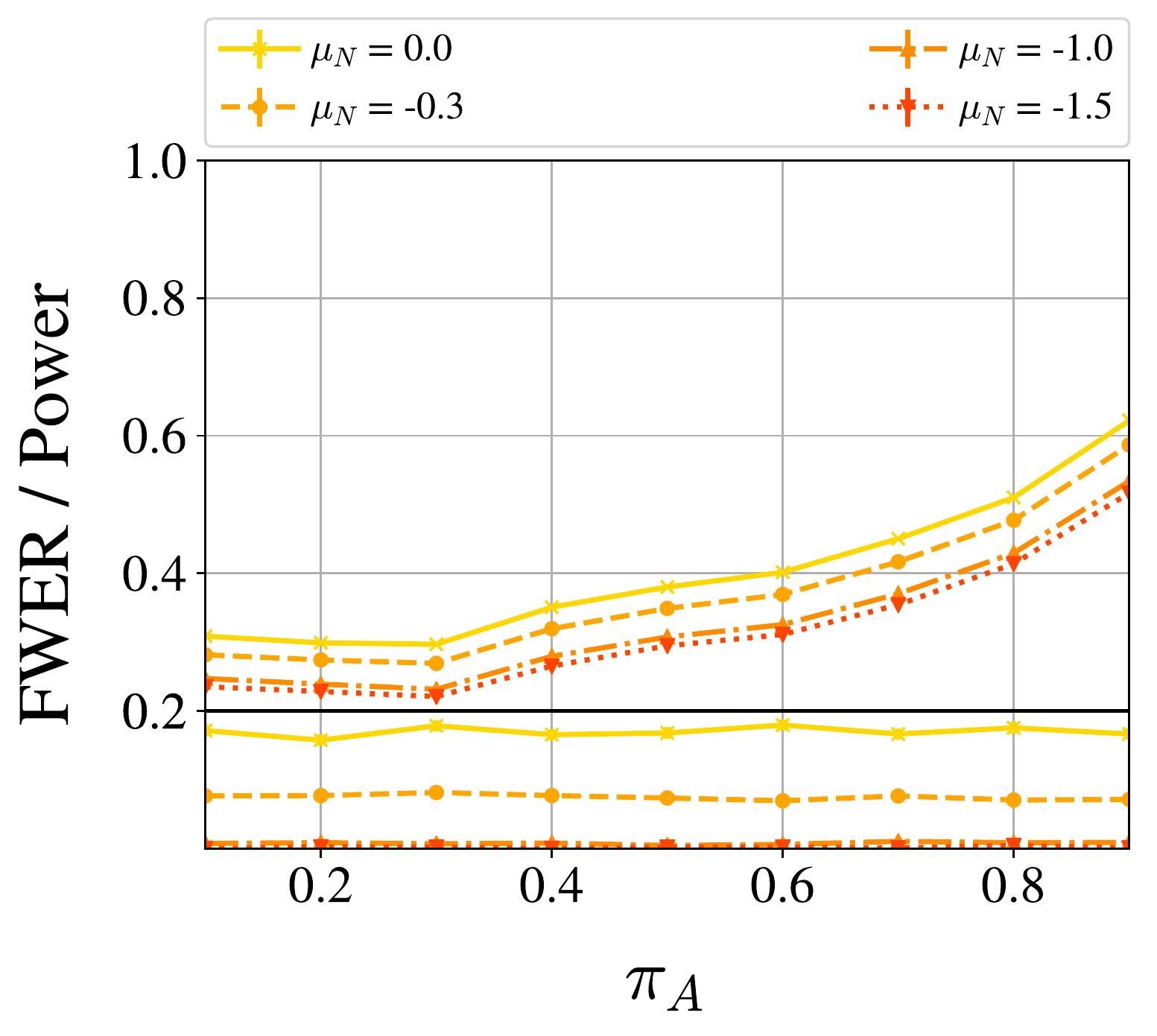}
		\caption{Statistical power and FWER versus fraction of non-null hypotheses $\pi_A$ for Adaptive-Spending, Alpha-Spending at target FWER level $\alpha= 0.2$ (solid black line). The curves above 0.2 line display the power of each methods versus $\pi_A$, while the lines below 0.2 display the FWER of each methods versus $\pi_A$. The $p$-values are drawn in the same way as described in \secref{simu}: we set $\mu_A = 4$ for both figures, but $\mu_N = 0$ for the left figure and $\mu_N \in \{0, -0.3, -1, -1,5\}$ for the right figure (hence the left nulls are conservative, the right nulls are not). These figures show that: (a) Adaptive-Spending do control FWER  at level 0.2; (b) Adaptive-Spending is more powerful than naive Alpha-Spending when there are many non-nulls (as shown in the left figure); (c) Adaptive-Spending loses power  when encounter many conservative nulls (as shown in the right figure). }\label{fig:adaptive}
	\end{figure}
	
	

	\subsection{Combining the two ideas: ADDIS-Spending, an adaptive and discarding algorithm}\label{sec:addis}
	From the discussion in \secref{discarding} and \secref{adaptive}, we find adaptivity and discarding both have their strength and weakness, however these are complementary. Therefore, we next combine those two ideas for a more sophisticated online FWER control method. Specifically, we present the following ADDIS-Spending algorithm, where ``ADDIS" stands for ``\underline{AD}aptive \underline{DIS}carding". In the following \secref{power} and \secref{numerical}, we demonstrate the power superiority of ADDIS-Spending over Alpha-Spending, with both theoretical justifications in \secref{more}, and numerical analysis using synthetic and real data in \secref{numerical}.
	
	\paragraph{\textbf{ADDIS-Spending}} Recall the definitions in and right after \eqnref{indicators}, we call any online FWER control method as ADDIS-Spending if it updates individual testing level $\alpha_i$ in a way satisfying the following conditions: (1) $\alpha_i$ is predictable, where the filtration $\F^{i-1} = \sigma(R_{1:i-1}, C_{1:i-1}, S_{1:i-1})$; (2) predictable sequences $\seq{\tau_i}$ and $\seq{\lambda_i} $ are such that $\lambda_i < \tau_i$ and $\alpha_i < \tau_i$ for all $i$ and 
	\begin{equation}\label{eq:addis-condition}
	\sum_{i \in \S \setminus \C} \frac{\alpha_i}{\tau_i-\lambda_i} \leq \alpha.
	\end{equation}


	Again we give a concrete example for ADDIS-Spending following similar logic in the previous two sections: choosing $\seq{\gamma_i }$ as a nonnegative sequence that sums to one, we test each $H_i$ at predictable level
	\begin{equation}\label{concreteaddis}
		\alpha_i := \alpha(\tau_i-\lambda_i) \gamma_{t(i)}, \quad \text{where } t(i) = 1+ \sum_{j < i}S_j -C_j.
	\end{equation}
	
	ADDIS-Spending can be regarded as a unification of the Adaptive-Spending and Discard-Spending algorithms we mentioned earlier: when setting $\lambda_i \equiv 0$ for all $i$, ADDIS-Spending recovers Discard-Spending; and when setting $\tau_i \equiv 1$ for all $i$, ADDIS-Spending recovers Adaptive-Spending. At a high level, the advantages of ADDIS-Spending come from this unification of adaptivity and discarding. 
	We obtained similar success in the ADDIS algorithm for online FDR control  \citep{tian2019addis} suggesting that our work may be regarded as a variant of ADDIS for a more stringent error metric.  
		Finally, we prove the FWER control of ADDIS-Spending. 
	\begin{theorem}\label{thm:addis}
		ADDIS-Spending controls the PFER (hence FWER) when null $p$-values are uniformly conservative as defined in \eqnref{conserve-def}, while being independent of each other and of the non-nulls.
	\end{theorem}
     \thmref{addis} is actually a special case of a more general version \thmref{local} that we are going to introduce later in \secref{local}, where we relax the independence assumption to local dependence, which includes independence as a special case. Therefore we omit the proof for the special case and present the proof for the more general version later.

	
	\subsection{Handling local dependence}\label{sec:local}	
	In \secref{addis}, we introduced ADDIS-Spending, a new powerful variants of Alpha-Spending, though at the cost of requiring independence among the $p$-values, while naive Alpha-Spending works even when $p$-values are arbitrarily dependent. Indeed, the assumption of independence is rarely met in real applications: tests that occur nearby in time may share the same dataset; null hypotheses are often constructed given the information of recent testing results; etc. On the other hand, arbitrary dependence between sequential p-values is also arguably unreasonable: the dataset used for testing or the testing results from the distant past is usually considered having no impact on the current testing. In light of this, we consider another dependence structure that is more realistic---local dependence, which is firstly proposed by \citet{zrnic2019asynchronous}, and is defined as the following:
	\begin{equation}\label{localdep}
		\textnormal{ For all } i \in \N, \textnormal{ there exists } L_i \in \mathbb{N},  \textnormal{ such that } P_i \perp P_{i-L_i-1}, P_{i-L_i-2}, \dots, P_1 ,
	\end{equation}
	where $\{L_i\}_{i=1}^{\infty}$ is a fixed sequence of parameters that we refer to as lags. We assume $L_{i+1} \leq L_i +1$, which is a reasonable requirement that the observable information does not decrease with time.  Implicitly, $P_i$ may be arbitrarily dependent on $P_{i-1},\dots,P_{i-L_i}$; and particularly, when $L_i \equiv 0 $ for all i, the local dependence reduces to the independence. We refer readers to the paper of \cite{zrnic2019asynchronous} for more detailed definition and discussions. 
	
	Here, we give simple alterations of the procedures in \secref{ADDIS-Spending} that allows them to deal with local dependence. The way we accomplish this is to follow the ``principle of pessimism'' \citep{zrnic2019asynchronous}. Specifically, this principle suggests ignoring what really happened in the previous $L_t$ steps when deciding what to do at time $t$, and hallucinate a pessimistic outcome for those steps instead. Formally, the alterations we made for procedures in \secref{ADDIS-Spending} insist that
	\begin{equation}
	    	\alpha_i, \tau_i, \lambda_i \in \F^{i-L_i-1}, \text{ for all } i \in \N,
	\end{equation}
	while still satisfying the other requirements in the corresponding original definitions. 
	
	As for concrete examples to implement the altered procedures described above, we present the altered concrete example of ADDIS-Spending in \eqref{concreteaddis}: we choose $\seq{\gamma_i }$ as an infinite nonnegative sequence that sums to one, we test each $H_i$ at predictable level
	\begin{equation}\label{concretelocal}
		\alpha_i := \alpha(\tau_i-\lambda_i) \gamma_{t(i)}, \quad\text{where } t(i)= 1+ L_i \land {(i-1)} + \sum_{j < i-L_i}S_j -C_j.
	\end{equation}
	Note that for this example, when $L_i = 0$ for all $i$, that is the local dependence structure reduces to independence, the modified procedures reduce to ADDIS-Spending in \eqref{concreteaddis} under independence.
	
	In the following, we specifically present the PFER (and hence FWER) control of altered ADDIS-Spending for local dependence in \thmref{local}, which is proved in \appref{pflocal}. 
	\begin{theorem}\label{thm:local}
		Altered ADDIS-Spending controls PFER (and hence FWER) in a strong sense when the null $p$-values are uniformly conservative as defined in \eqnref{conserve-def} and follow the local dependence defined in \eqnref{localdep}.
	\end{theorem}

\section{Statistical Power}\label{sec:power}
    In this section, we study the statistical power of Alpha-Spending and ADDIS-Spending under an idealized Gaussian setting with randomly arriving signals. Specifically, in \secref{spending} we examine the power of Alpha-Spending, and we derive some optimal choices of the underlying $\seq{\gamma_i}$ in a certain range, for either fixed or varying signal strength and density. Then we provide theoretical justification for the benefits of adaptivity and discarding in \secref{more}: we prove that if $\seq{\gamma_i}$ lies in the aforementioned optimal range, ADDIS-Spending is more powerful than Alpha-Spending.
	
	Before we proceed with the analysis of power, it is useful to set up a few definitions. 
	\begin{definition}\label{def:q-series}($q$-series and log-$q$-series) For any $q > 1$, we call an infinite sequence $\seq{\gamma_i}$ which is nonnegative and sums to one as $q$-series if $\gamma_i \propto i^{-q}$ for all $i$, and similarly, as a log-$q$-series if $\gamma_i \propto \frac{1}{(i+1)\log^{q}{(i+1)}}$ for all $i$. 
	\end{definition}
	
	\begin{definition}\label{def:meantest}
		(Gaussian mean testing problem) We call the problem of testing a possibly infinite sequence of hypotheses $\seq{H_i}$ as Gaussian mean testing problem, if, each observation $Z_i$ follows the following mixed distribution:
		\begin{equation*}\label{mixguassian}
			Z_i = 
			\begin{cases}
				X_i + \mu_{A} & \text{with probability } \pi_{A};\\
				X_i + \mu_{N}, & \text{with probability } 1-\pi_{A},
			\end{cases}
		\end{equation*}
        where constants $\mu_{A}>0,\ \mu_{N}\leq 0,\ \pi_{A}\in (0,1)$ for all $i$, $\ X_i \stackrel{iid}{\sim} N(0,1)$, and we test the null hypothesis $H_i: \mu_i \leq 0$, where  $\mu_i= \EE{Z_i}$. 
	\end{definition}
	
	\noindent
Recalling \secref{discarding}, if the $p$-values calculated are one-sided that is $P_i = \Phi^{-1}(-Z_i)$, then we know that the nulls are uniformly conservative as defined in \eqnref{conserve-def}, strictly conservative when $\mu_{N}<0$, and uniform when $\mu_{N}=0$. 
	
    In the following, we only consider the online Gaussian mean testing problem described above. 
    We compare the algorithms that are presented as concrete examples of each method, which are formulas \eqref{spend}, \eqref{concreteaddis} for Alpha-Spending and ADDIS-Spending respectively, with the same underlying sequence $\seq{\gamma_i}$. Also, for simplicity, we use the number of true discoveries $D$ as one of the performance measures. Noted that $D$ is the numerator inside the expectation of the power function \eqref{fdr}. Since the denominator inside the expectation of the power function in \eqnref{fdr} remains the same for different algorithms given the same testing sequence, the expectation of numerator $D$ may arguably serve as a nice substitution for power function in respect of comparison. Hence we refer $\EE{D}$ also as the power of online FWER control methods in this section. 
	
	
	\subsection{Getting optimal power using naive Alpha-Spending}\label{sec:spending}
	In this section, we derive optimal choices of $\seq{\gamma_i}$ in the range of $q$-series for Alpha-Spending with regard the Gaussian mean testing problem in \defref{meantest}. As we discussed before, the expectation of number of true discoveries serves as a reasonable measurement for comparing the power of testing procedures. Recall that $D$ is the number of true discoveries. In the Gaussian mean testing problem, we have
	\begin{align}\label{ED}
		\EE{D}& = \EE{\sum_{i=1}^{\infty} \one{P_i\leq \alpha_i , i \in \nullsc}} \stackrel{(i)}{=} \sum_{i=1}^{\infty} \EE{\one{P_i \leq \alpha_i , i \in \nullsc}}\nonumber\\& = \sum_{i=1}^{\infty} \PP{P_i\leq \alpha_i , i \in \nullsc} =  \sum_{i=1}^{\infty}\pi_{A}\Phi(\Phi^{-1}(\alpha_i) + \mu_{A}) = \sum_{i=1}^{\infty}\pi_{A}\Phi(\Phi^{-1}(\alpha\gamma_i) + \mu_{A}),
	\end{align}
	where $\Phi$ is the standard Gaussian CDF, and (i) is true due to the fact that each entry in the summation is nonnegative, and the last step uses the fact that $\alpha_i = \alpha\gamma_i$ for all $i$ in Alpha-Spending. Additionally, for each $N \in \mathbb{N}$, we denote $\EEN{D}$ as the expectation of true discoveries among the first $N$ hypotheses, which means 
	\begin{equation}\label{END}
		\EEN{D} \defn \sum_{i=1}^{N}\pi_{A} \Phi(\Phi^{-1}(\alpha \gamma_i) + \mu_{A}).
	\end{equation}
	
	\noindent
    It is obvious that the power of Alpha-Spending does not depend on $\mu_{N}$, therefore we only consider how to choose $\seq{\gamma_i}$ to optimize the power given different  $\mu_{A}$ and $\pi_{A}$. Specifically, we derive that for the Gaussian mean testing problem in \defref{meantest} using Alpha-Spending, the optimal sequence $\seq{\gamma_i}$ in the range of  $q$-series will be  $q=1^{+}$ for any choice of $\mu_A>0$ and $\pi_A \in (0,1)$. This result provides some heuristic of choosing $\seq{\gamma_i}$: one should resort to log-$q$-series for higher power when applicable. The formal results are stated in the following \thmref{fixed}.

	
	
	\begin{theorem}\label{thm:fixed}
		Recall the definition of $\EE{D}$ and  $\EEN{D}$ in equations \eqref{ED} and \eqref{END}. For Alpha-Spending \eqref{spend} at level $\alpha < 1/2$, if the underlying sequence $\seq{\gamma_i}$ is a $q$-series where $q>1$, then for the Gaussian mean testing problem in \defref{meantest}, 
		we have
		\begin{itemize}
			\item[(a)]$\EEN{D}$ is a function increasing with $q$ first and then decreasing with $q$ for $N\geq 2$. Additionally, defining 
			\begin{equation}
				q^{*}(N, \mu_A) \equiv \textnormal{argsup}_{q>1} \EEN{D},
			\end{equation}
			we have that $q^{*}(N, \mu_{A})$ is monotonically decreasing with $N$ for any $\mu_A > 0$, and
			\begin{equation}
				\lim_{N \to \infty} q^{*}(N, \mu_A) = 1.
			\end{equation}
			
			\item[(b)] $\EE{D}$ is finite for any fixed $q$, and is a function monotonically decreasing with $q$. 
		\end{itemize}
	\end{theorem}
	\noindent
	\thmref{fixed} in fact suggests that the slower $\seq{\gamma_i}$ is decaying, the higher the power will get, which corresponds to  our intuition that the power will be higher if we protect the ability to detect the signals in the long run, or we try not to run out of our total budget too fast. However, if the testing process stops at a certain point, then the sequence $\seq{\gamma_i}$ that decays too slow will hurt the power, for the testing will not get the benefit from the long run. This trade-off implies that there is an optimal sequence, which is not decaying too slow or too fast, making the testing process achieve the highest power. \thmref{fixed} provides theoretical verification for those intuitions and is proved in \appref{fixedpf}.
	
	Therefore, when we have no prior information on the hypotheses, which means we could only treat the non-null fraction and the non-null mean as some fixed arbitrary value, we could always resort to the sequence $\seq{\gamma_i}$ that sums to one with slower decay rate to obtain higher power of Alpha-Spending. For example, among $q$-series, we should choose  $q$ as close to one as possible, and one should resort to log-$q$-series for higher power.
	
     The above results are all in the regime of fixed signal strength and density, which is a bit unrealistic in practice, though potentially suitable for deriving interpretable heuristics. For completeness, we also go beyond the setting of fixed signal strength and density: in \appref{varied}, we consider different $\mu_{Ai}$ and $\pi_{Ai}$ for each $i$ in \defref{meantest}, and we show that, if $\seq{\mu_{Ai}}$ and $\seq{\pi_{Ai}}$ satisfy some reasonable conditions, then there exists a function $h$ with closed form, such that $\gamma_i = h(\pi_{Ai}, \mu_{Ai})$ achieves the highest power. We refer the reader to the \appref{varied} for  details.
	
	 \subsection{The adaptive discarding methods are more powerful}\label{sec:more}
	In \secref{spending}, we showed that the optimal choice of $\seq{\gamma_i}$ for Alpha-Spending in the range of $q$-series and log-$q$-series for the Guassian mean testing problem with fixed signal strength and density lies in the regime of log-$q$-series. Here, we show that in this  regime, ADDIS-Spending is provably more powerful than  Alpha-Spending.
	For simplicity, we consider fixed $\tau_i$ and $\lambda_i$, that is $\tau_i = \tau$ and $\lambda_i = \lambda$ for all $i \in \N$, and we denote the number of discoveries from  ADDIS-Spending as $D_{\textnormal{ADDIS}}(\lambda, \tau)$, and $D_{\textnormal{spend}}$ for Alpha-Spending. Below, we demonstrate that as long as the hyper-parameters $\lambda$ and $\tau$ are reasonably chosen, ADDIS-Spending is guaranteed to be more powerful than Alpha-Spending. 
	
	
		
	\begin{theorem}\label{thm:more}
		For the Gaussian mean testing problem in \defref{meantest}, 
		if the underlying $\seq{\gamma_i}$ is a log-$q$-series as defined in \defref{q-series} with $q>1$, then there exists some $c^{*}$ such that with probability one, 
		\begin{equation}
		    \EE{D_{\textnormal{ADDIS}}(\lambda, \tau)} \geq \EE{D_\textnormal{spend}} \quad \textnormal{for all }\ \lambda\in [0,c^{*}) \text{ and all } \tau \in (c^{*},1], \textnormal{ s.t. } \tau-\lambda < 1.
		\end{equation}
		Additionally, $c^{*}$ increases with $\pi_A$, $\mu_A$ and $\mu_N$, and particularly reaches one when $\mu_N = 0$.
	\end{theorem}	 
	The proof is in \appref{pfmore}.
    \thmref{more} indicates that as long as we have reasonable prior information about the signal strength, density and conservativeness of nulls, we can utilize them to achieve higher  power over Alpha-Spending via appropriate choice of the two hyperparameters $\lambda$ and $\tau$ in ADDIS-Spending. In our experiments, $\lambda=1/4$ and $\tau=1/2$ perform just fine as we discuss next, and we suggest these as defaults. As \thmref{more} provides the theoretical justification for the benefits of discarding and adaptivity in terms of power, in the following \secref{numerical} we provide numerical analysis with both simulations and real data example to confirm these benefits. 
	
	\section{Numerical Studies}\label{sec:numerical}
	\subsection{Simulations}\label{sec:simu}
	In this section, we provide some numerical experiments to compare the performance of ADDIS-Spending, Discard-Spending, Adaptive-Spending, and Alpha-Spending. 
	In particular, for each method, we provide empirical evaluations of its power while ensuring that the FWER remains below a chosen value. Specifically, in all our experiments, we control the FWER under $\alpha = 0.2$ and estimate the FWER and power by averaging over 2000 independent trials. The constant sequences $\tau_i \equiv 1/2$ in Discard-Spending, $\lambda_i \equiv 1/2$ in Adaptive-Spending, and $\lambda_i \equiv  1/4, \tau_i \equiv  1/2$  in ADDIS-Spending for all $i \in \N$ were found to be generally successful, so as our default choice in this section and we drop the index for simplicity. Additionally, we choose the infinite sequence $\gamma_i \propto 1/(i+1)\log(i+1)^2$
	for all $i \in \N$ as default, which could be substituted by any constant infinite sequence that is nonnegative and sums to one. Those default choices turn out working pretty well in establishing the strength of our methods and confirm the theoretical results in \secref{more}, though they may not optimal in obtaining high power. 
	
	In what follows, we show the power superiority of ADDIS-Spending over all other three methods, especially under settings with both nonnegligible number of signals and conservative nulls. Specifically, we consider the simple experimental setup of Gaussian mean testing problem in \defref{meantest} with $T=1000$ components, where the nulls are uniformly conservative from the discussion in \secref{discarding}.
     We ran simulations for $\mu_N \in \{0, -0.5, -1, -1.5\}$, $\mu_A \in \{4, 5\}$, and $\pi_A \in \{0.1,0.2,\dots,0.9\}$, to see how the changes in  conservativeness of nulls and true signal fraction may affect the performance of algorithms. The results are shown in \figref{addis}, which indicates that (1) FWER is under control for all methods in all settings; (2) ADDIS-Spending enjoys appreciable power increase as compared to all the other three methods in all settings; (3) the more conservative the nulls are (the more negative $\mu_N$ is), or the higher the fraction of non-nulls is, the more significant the power increase of ADDIS-Spending is.
	
	\begin{figure}
		\begin{subfigure}{0.32\textwidth}
			\centering
			\includegraphics[width=\linewidth]{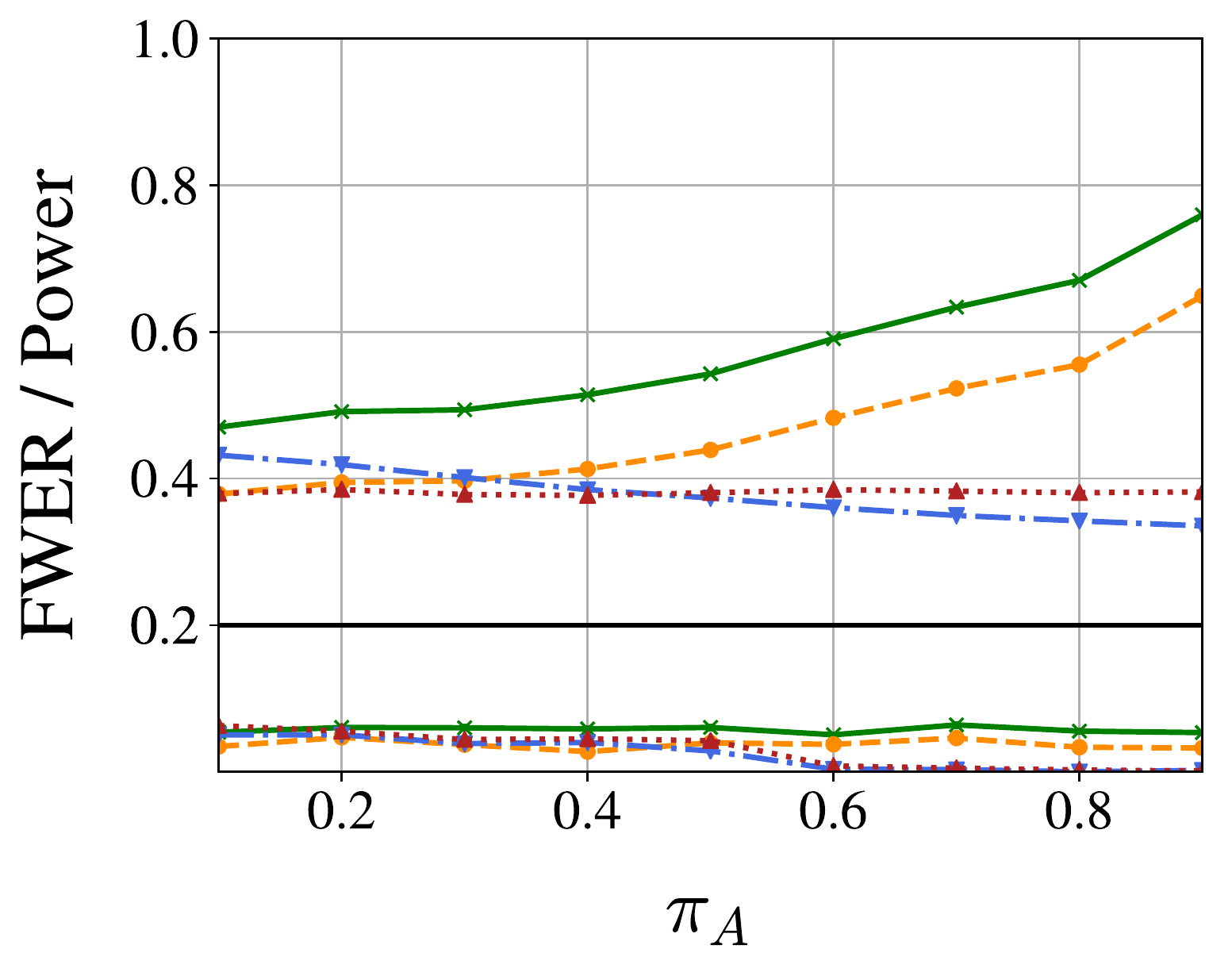}
			\caption{}
		\end{subfigure}%
		\begin{subfigure}{0.32\textwidth}
			\centering
			\includegraphics[width=\linewidth]{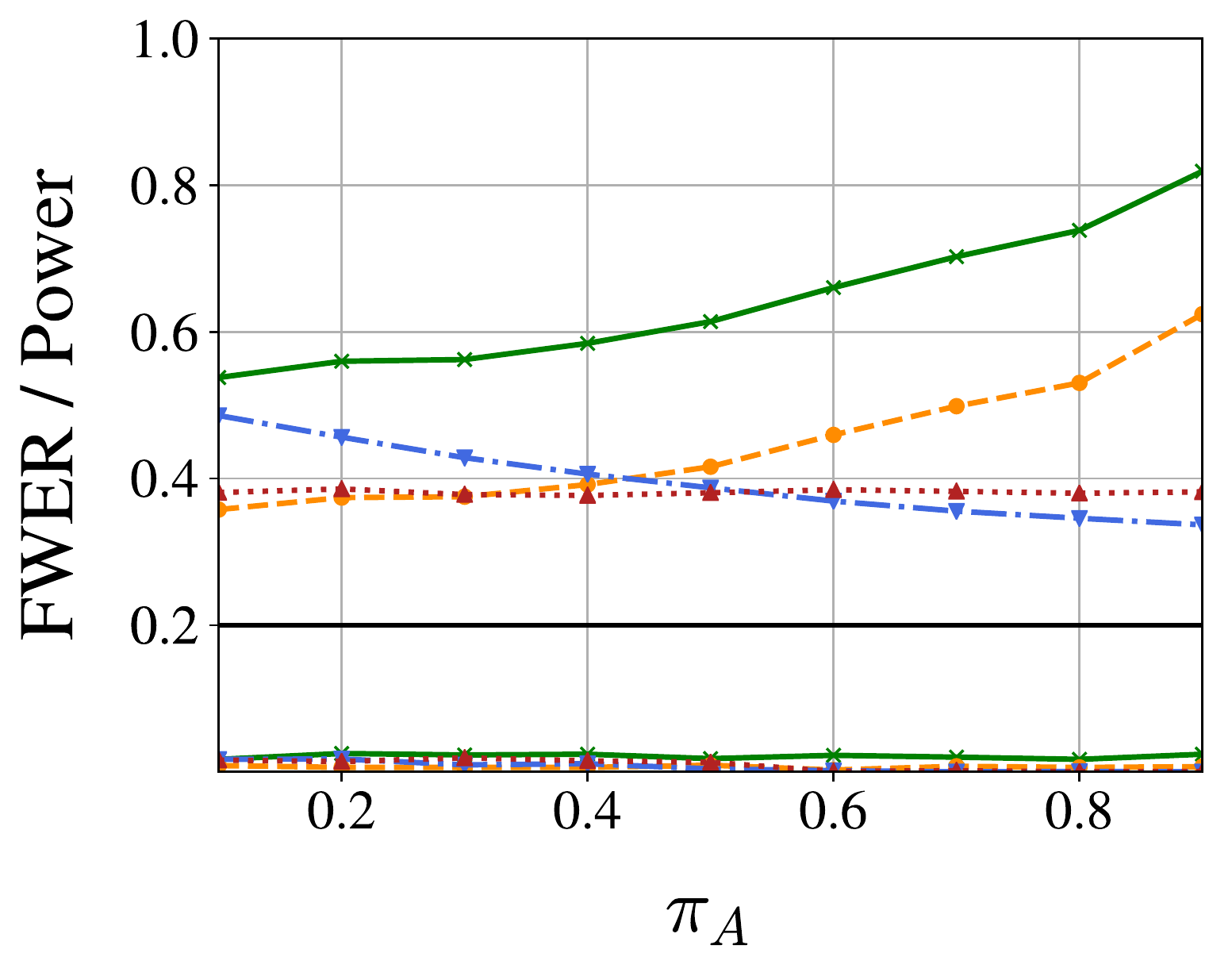}
			\caption{}
		\end{subfigure}	%
		\begin{subfigure}{0.32\textwidth}
			\centering
			\includegraphics[width=\linewidth]{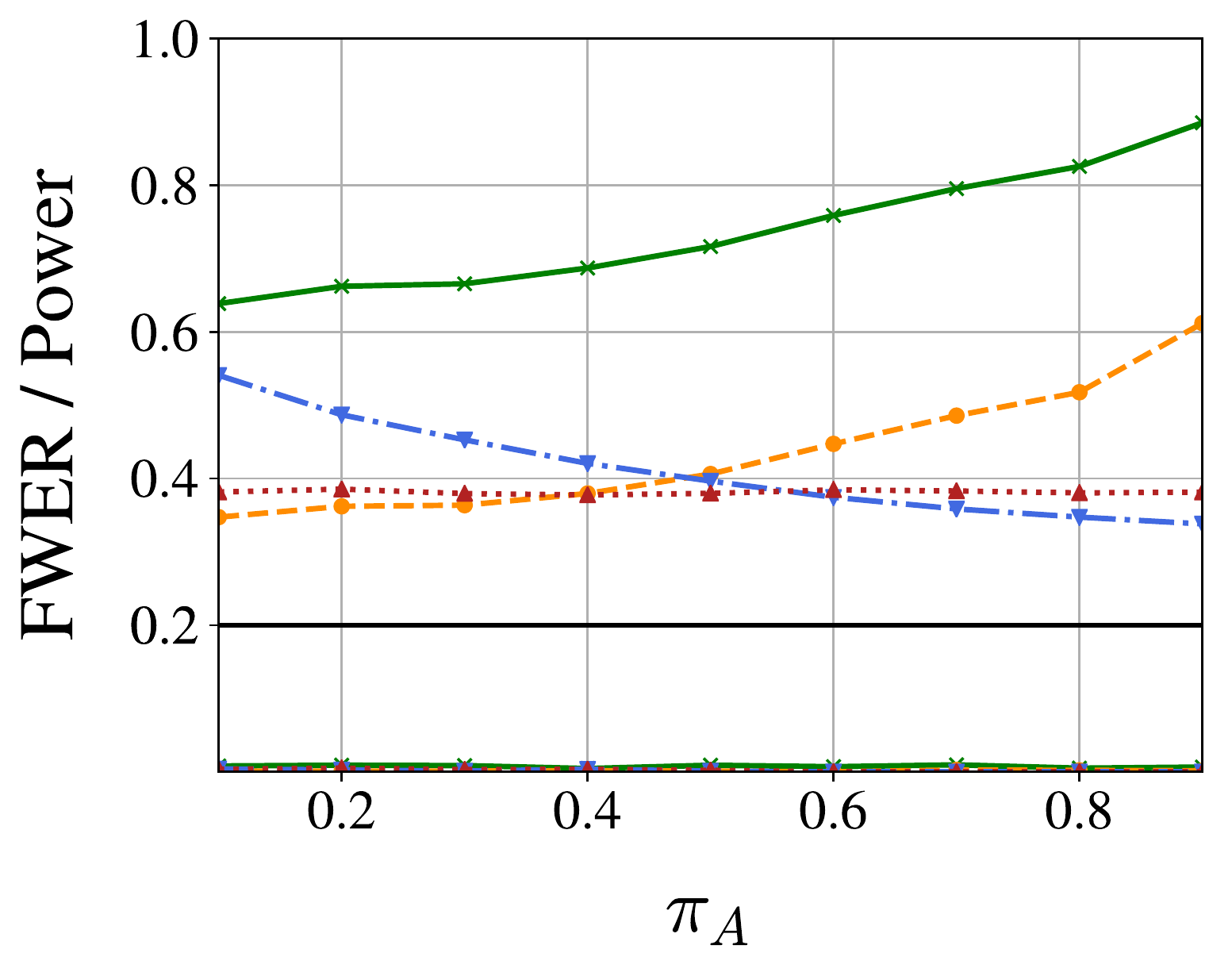}
			\caption{}
		\end{subfigure}\\
		
		\begin{subfigure}{0.32\textwidth}
			\centering
			\includegraphics[width=\linewidth]{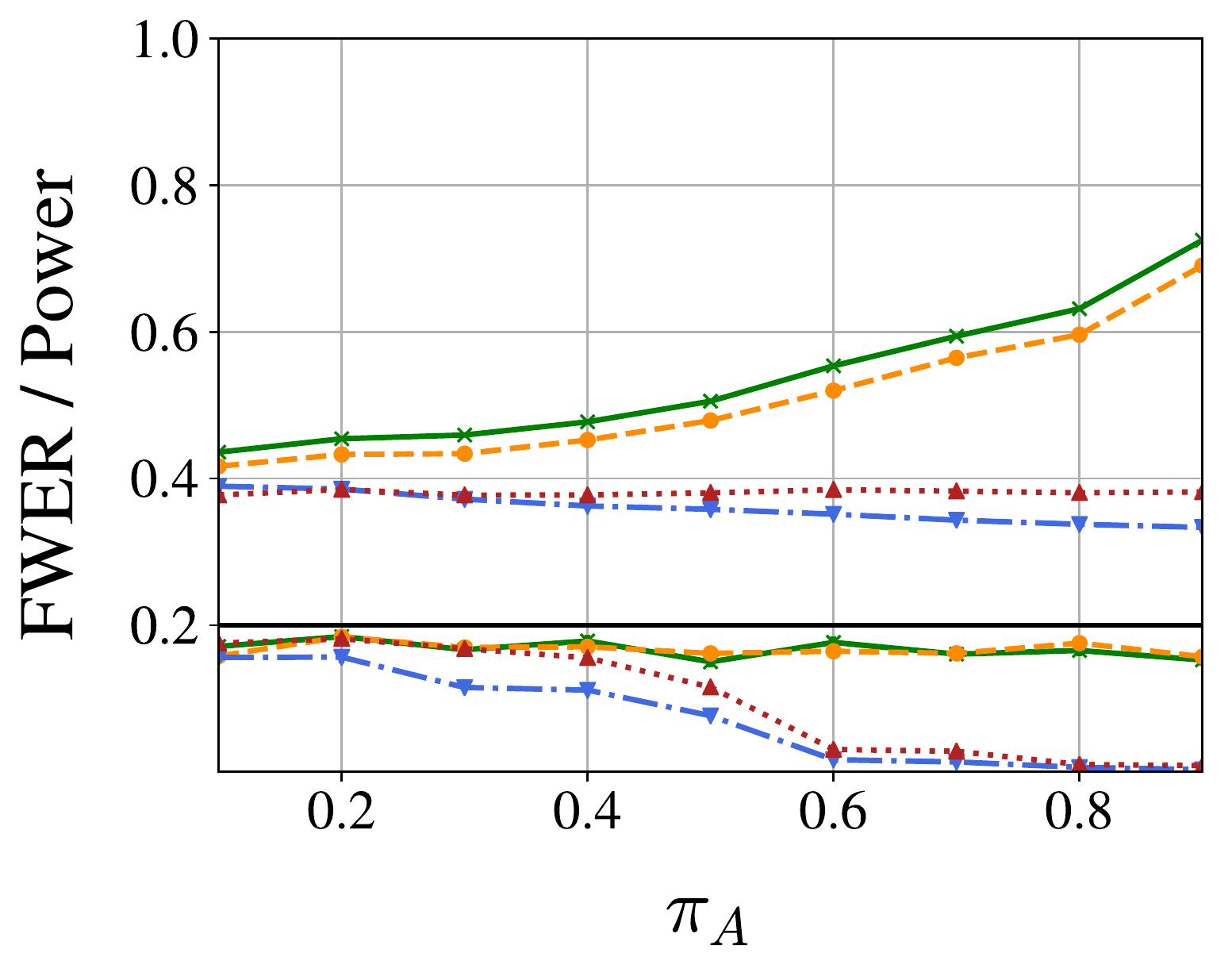}
			\caption{}
		\end{subfigure}%
		\begin{subfigure}{0.32\textwidth}
			\centering
			\includegraphics[width=\linewidth]{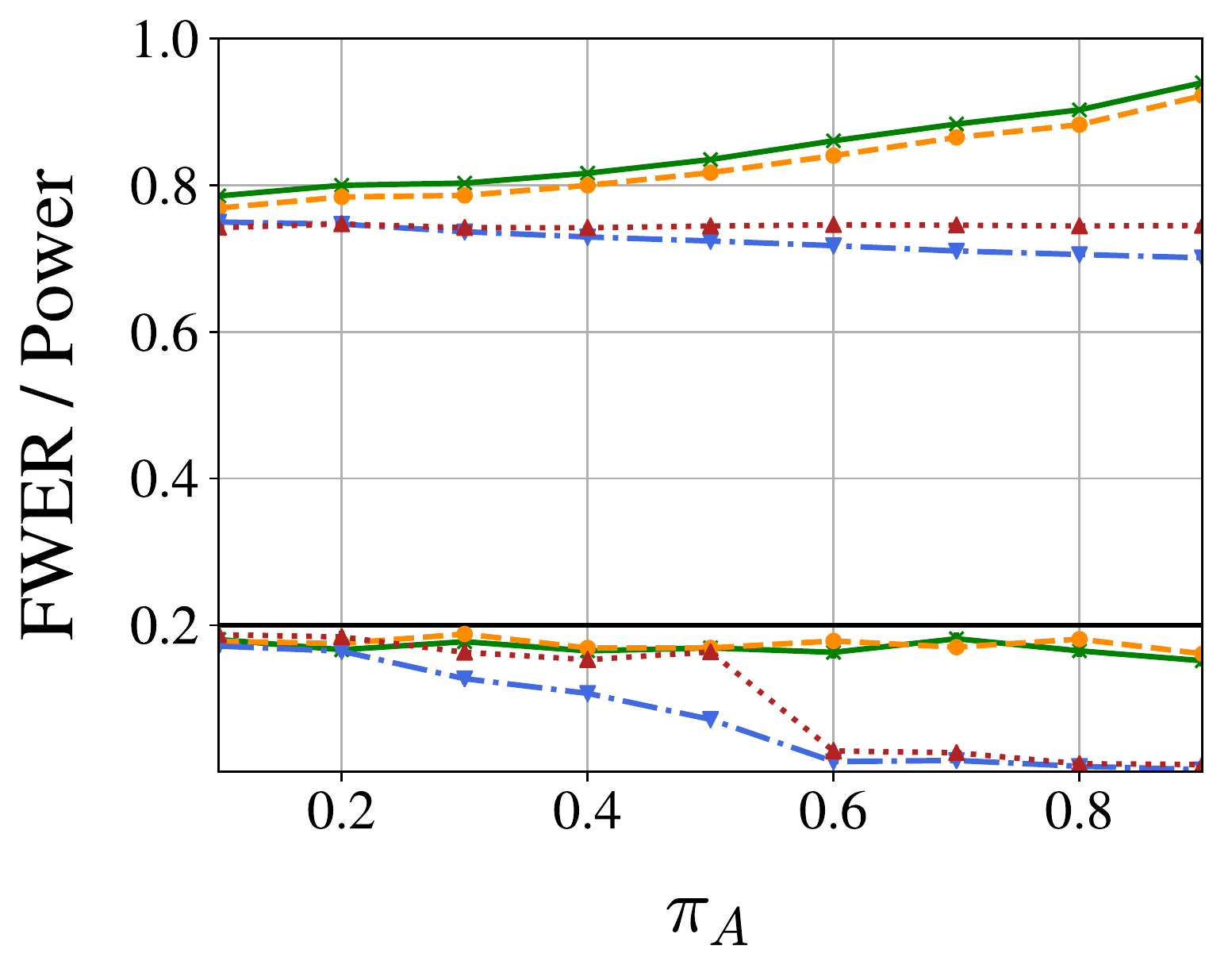}
			\caption{}
		\end{subfigure}
		\includegraphics[width=0.17\linewidth, valign=c]{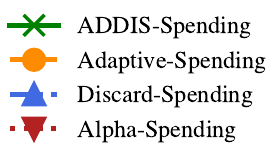}
		\caption{Statistical power and FDR versus fraction of non-null hypotheses $\pi_A$ for ADDIS-Spending, Discard-Spending, Adaptive-Spending and Alpha-Spending at target FWER level $\alpha$ = 0.2 (solid black line). The lines above the solid black line are the power of each methods versus $\pi_A$, and the lines below are the FWER of each methods versus $\pi_A$. The $p$-values are drawn using the Gaussian model as described in the text, while we set $\mu_N = -0.5$ in plot (a), $\mu_N = -1$ in plot (b), $\mu_N = -1.5$ in plot (c), and $\mu_N = 0$ in plots (d) and (e); and we set $\mu_A = 4$ in plots (a), (b), (c), (d), $\mu_A = 5$ in plot (e). Therefore nulls in (a), (b) and (c) are conservative, and the conservativeness is increasing; while nulls in (d) and (e) are not conservative (uniform). }\label{fig:addis}
	\end{figure}
	
	\subsection{Application to the International Mouse Phenotyping Consortium (IMPC)}\label{sec:real}
	 In the following, we introduce the application of ADDIS-Spending to real-life data about International Mouse Phenotyping Consortium (IMPC). The IMPC coordinates a large study to functionally annotate every protein coding gene by exploring the impact of the gene knockout on the resulting phenotype. Statistically speaking, for each phenotype $i$, people test the null hypothesis $H_j^{(i)}$: ``the knockout of gene j will not change the phenotype i" versus its alternative, via comparing the unmutated mouse (control case) to the mouse with gene j knockout. Since the dataset and resulting hypotheses constantly grow as new knockouts are studied, and a positive test outcome could lead to following up medical research that cannot be revised, therefore it is natural to view this as an online hypothesis testing problem as the ones considered in the previous sections. 
    
    We follow the analysis done by \citet{karp2017prevalence}, which resulted in a set of $p$-values for testing genotype effects. The data is available at Zenodo repository \url{https://zenodo.org/record/2396572}, organized by \citet{robertson2019onlinefdr}. This particular dataset admits a natural local dependence structure: the hypotheses are tested in small batches, with each batch using a different group of mice. Plot (a) in \figref{real} demonstrates this local dependence structure via showing the first 5000 $-log_{10}$ transformed $p$-values:  the transformed $p$-values are ordered by the time that its corresponding  data samples are collected, and the adjacent batches are distinguished using different colors.
    
    Due to this online nature and local dependence structure of the data, we hence apply the modified version of ADDIS-Spending described in \secref{local}. Since we do not know the underlying truth, therefore we only report the number of discoveries and argue that the corresponding FWER are all under control following our theoretical results. Plot (b) in \figref{real}
     shows the power advantage of ADDIS-Spending over Alpha-Spending and Online-Fallback, where we use the same underlying $\seq{\gamma_i}$ sequence with $\gamma_i \propto 1/i^{1.1}$, and the default setting $\tau = 0.5, \lambda = 0.25$ for ADDIS-Spending in comparison.
    
    The above real data example again supports our key idea: utilizing the independence or local dependence structure to incorporate with adaptability and discarding can improve the power of online testing procedure much.
    
    \begin{figure}
    \centering
		\begin{subfigure}{0.5\textwidth}
			\centering
			\includegraphics[width=\linewidth]{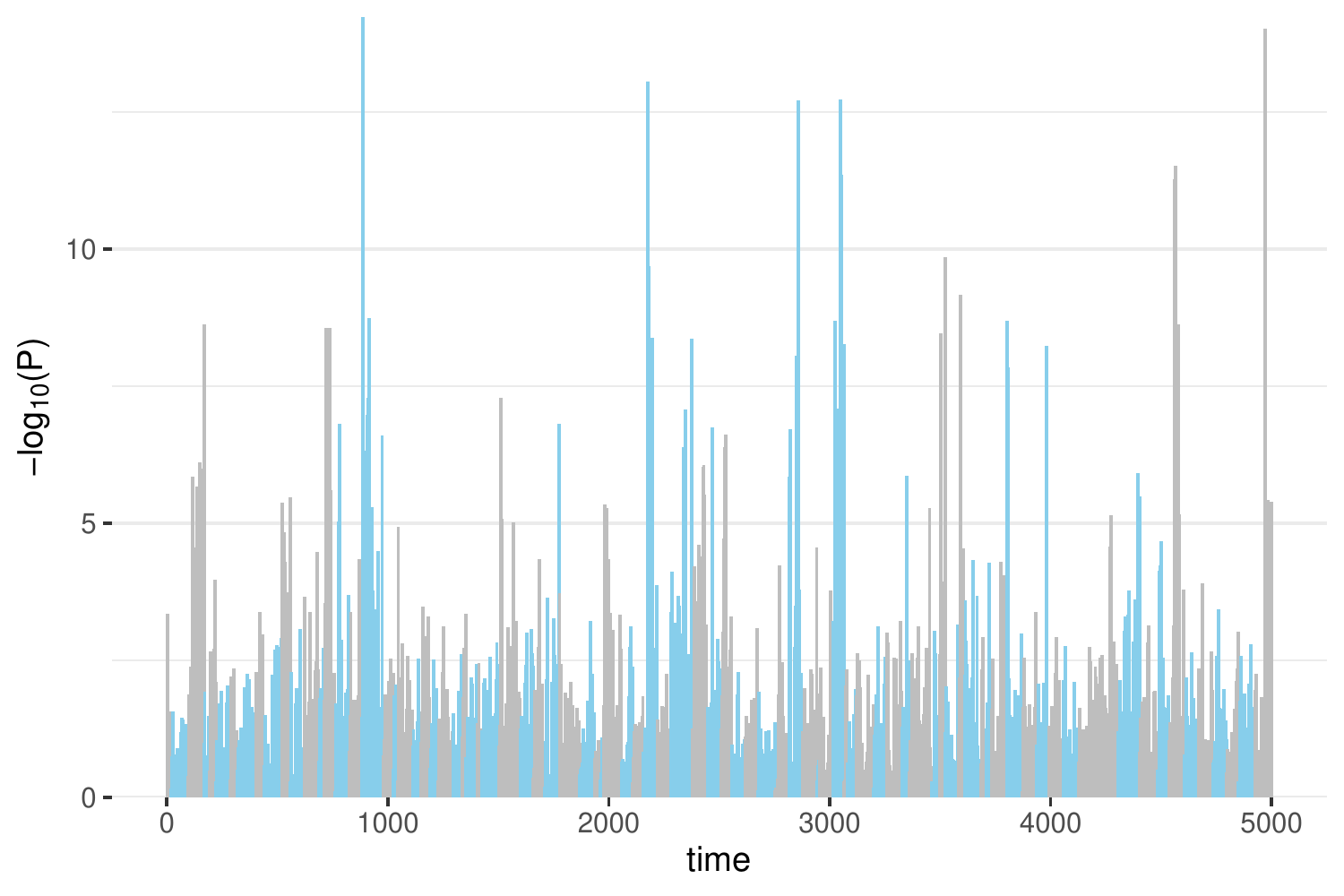}
			\caption{}
		\end{subfigure}%
		\begin{subfigure}{0.5\textwidth}
			\includegraphics[width=\linewidth]{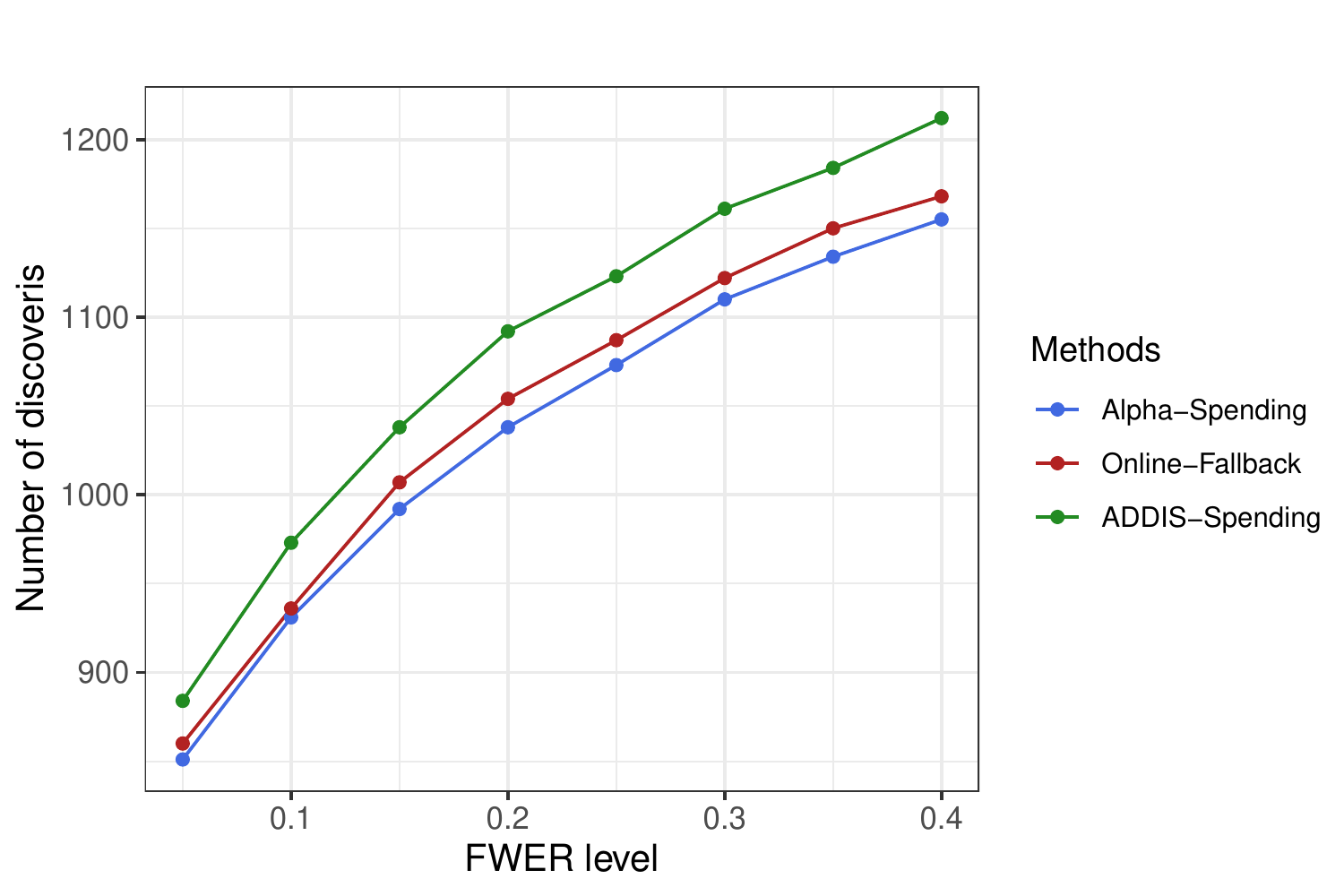}
			\caption{}
		\end{subfigure}	
		\caption{Figure (a) shows the barplot of $-\log_{10}$ transformed $p$-values, where the $p$-values are ordered in time, and each adjacent batches are distinguished with different colors. Figure (b) shows the number of discoveries versus FWER level, using different algorithms (Alpha-Sepnding, Online-Fallback, ADDIS-Spending).}\label{fig:real}
	\end{figure}
	

	\section{Conclusion and Discussion}\label{sec:conclude}
	Modern biology studies often require testing hypotheses in a sequential manner, and how to control familywise error rate in this setting leads to the statistical problem of online FWER control. This paper derives new algorithms for online FWER control, a problem for which no systematic treatment exists in the literature to the best of our knowledge. While we describe several new methods, each improving on Alpha-Spending (online Bonferroni) in different ways, the most promising of these in experiments seems to be ADDIS-spending, a new adaptive discarding algorithm that adapts to both unknown number of non-nulls and conservativeness of the nulls. 
	
     Though we find that ADDIS-Spending is the most promising method within current practices, we are also wondering whether there exists universal refinements over ADDIS-Spending. It is known that for any offline global null testing methods, there exists a closure such that the power of it is unimprovable. So we are wondering whether there exists similar logic for online multiple testing. We provide some initial attempts on developing the variant of the closure principle for online multiple testing in \appref{closedapp}. Several questions still remain: Is our proposed principle essentially unimprovable? Also, is the closure of an online method still be an online method? If not always, then what are the cases in which it is? We leave these as open questions for future work. 
     
     The application of adaptivity and discarding go much beyond the main contribution of this paper: they can also be applied to methods in \secref{current} to develop more powerful variants. We provide some concrete examples and corresponding proofs for their FWER control in \appref{addisapp} for interested readers.
		 
	 Throughout this paper, we mainly discuss the online algorithms for controlling FWER. In real applications, many prefer to control $k$-$\FWER$ instead, in order to obtain a less stringent error control.  The $k$-$\FWER$ is defined as $\PP{ V \geq k}$, which reduces to FWER as $k=1$. It is straightforward that for any methods that have PFER control, changing the sum of the test levels to $k\alpha$ will assure $k$-FWER controlled at level $\alpha$, simply using Markov's equality. Therefore, all our new algorithms that provably have PFER control may easily be extended to $k$-FWER control methods.

   \section{Code and data availability}
    The code to reproduce all figures in this paper are accessible at \url{https://github.com/jinjint/onlineFWER}, and the real data set used in the paper is available at Zenodo repository \url{https://zenodo.org/record/2396572}, organized by \citet{robertson2019onlinefdr}. Additionally, an R package called \texttt{onlineFDR} \citep{online} developed by David Robertson and the authors of this paper (among others), contains current state of arts in all aspect of online multiple testing, including online FWER control and also the new algorithms proposed here. 
    
 \section*{Acknowledgments}
The authors thank Jelle Goeman and David S. Robertson for a careful reading and useful suggestions, particularly pointing the authors to the related work of \cite{goeman2010sequential, ellis2017gaining}, and available real data set IMPC.

	\bibliographystyle{unsrtnat}
   \bibliography{refs}	

\renewcommand{\thesection}{\Alph{section}.\arabic{section}}
\setcounter{section}{0}





\newpage

\begin{appendices}
		
\section{Proof of \propref{recycling}}\label{app:recycling}
		The Online Fallback procedure is in fact equivalent to the following procedure, which may be regarded as an extension of the graphical procedure proposed by  \citet{bretz2009graphical} to the infinite case. Specifically, recalling the form of Online Fallback in \eqref{recycling} for testing hypotheses set $\{H_i\}_{i=0}^{\infty}$, we denote $\alpha_i = \alpha \gamma_i$ for all $i$, and 
		\begin{equation}\label{W}
			W = \begin{bmatrix}
				0     & w_{12} &  w_{13}&  \dots   & w_{1i} & \dots \\
				&     0    & w_{23} & w_{24} & \dots    & \dots       \\
				&          &    0     & w_{34} & \dots    & \dots        \\
				&          &          & \ddots   &    & \vdots \\
				&          &          &          & \ddots         & \vdots\\
			\end{bmatrix}.
		\end{equation}
		Note that the nonnegative $\{\alpha_i\}_{i=0}^{\infty}$ has summation no more than $\alpha$, while the nonnegative sequence $\{w_{ki}\}_{i=k+1}^{\infty}$ has summation less than one for each $k \in \mathrm{N}$. Now we consider testing the first $n$ hypotheses $\{H_i\}_{i=1}^{n}$ for arbitrary number $n$, using the graphical procedure Algorithm 1 in \citep{bretz2009graphical} with the initial significance levels $\{\alpha_i\}_{i=0}^{n}$, and transfer graph $W_n$, which takes the first $n$ rows and $n$ columns from $W$. Then the graphical procedure is the following:
		\begin{itemize}
		    \item{}Step 0: Set $M = \{1,2,\dots,n\}.$
		    \item{}Step 1: For the smallest $t \in M$ that $P_t \leq \alpha_t$, we reject $H_t$, and update $M = M \setminus \{t\}$, $\alpha_j = \alpha_j + \alpha_t w_{tj}$ for all $j \in M$, $w_{jk} = \frac{w_{jk} + w_{jt}w_{tk}}{1-w_{jt}w_{tj}}$ for all $j \neq k > t$.  If such $t$ does not exist, stop.
		    \item{}Step 2: If $|M| >1$, go to step 1; otherwise stop. 
		\end{itemize}
		Due to the upper triangle structure of the transfer graph, we have $w_{tk} \equiv 0$ for $k>t$ in the updating step in Step~1 above, therefore the updating rules with regard $W_n$ reduces to $w_{jk} = w_{jk}$, which means the transfer graph $W_n$ remains the same during the testing process. Due to the upper triangle structure of our graph, the significance level of the tested hypotheses will remain the same whenever we perform updating. Therefore, we do not need to re-evaluate the hypotheses we already checked (i.e. $H_1,\dots, H_t$), and could proceed the testing in the original sequential order. These facts imply that this graphical procedure based on $W_n$ is equivalent to testing the first $n$ hypotheses using the following procedure: 
		\begin{itemize}
		    \item[(*)] For $t = 1, 2, \dots, n$, we reject $H_t$ if $P_t \leq \alpha_t$, and update the significance level for the future test with $\alpha_j = \alpha_j + \alpha_t w_{tj}$ for all $j > t$, and proceed with the testing.
		\end{itemize}
		Note that, this procedure (*) always achieves the same rejection set with Online Fallback when testing the first $n$ hypotheses for any $n \in \N$, and it exactly recover Online Fallback when $n$ goes to infinity. 
		
		Based on the above discussion, we establish the FWER control of Online Fallback in the following. Denote $A_n:=\{V_n = 0\}$, and $V_n \defn \sum_{i=0}^{n}\one{P_i\leq \alpha_i, i \in \nulls}$ for all $n\in \N$, we can write the FWER for Online Fallback as the following:
		\begin{align}\label{pfrecycling}
			\textnormal{FWER} &= \PP{V \geq 1} = \PP{\sum_{i=0}^{\infty}\one{P_i\leq \alpha_i, i \in \nulls}\geq 1}\nonumber\\
			& = 1 - \PP{\sum_{i=0}^{\infty}\one{P_i\leq \alpha_i, i \in \nulls} = 0} = 1 - \PP{\lim_{n\to \infty} \{ \sum_{i=0}^{n}\one{P_i\leq \alpha_i, i \in \nulls} = 0\}} \nonumber\\
			& = 1 - \PP{\lim_{n\to \infty} A_n} \stackrel{(i)}{=} 1 - \lim_{n\to \infty} \PP{A_n},
		\end{align}
		where (i) is true since $A_1\supseteq A_2\supseteq \dots A_n \supseteq \dots$ is a monotonically decreasing sequence. As discussed above, for testing the first $n$ hypotheses, Online Fallback is equivalent to procedure (*), which is in turn equivalent to the graphical procedure based on $W_n$ and $\{\alpha_i\}_{i=0}^{n}$. Since the sum for each row in the transfer graph $W_n$ is less than one, and the updated significance level is always larger than the original one, therefore, the graphical procedure based on $W_n$ and $\{\alpha_i\}_{i=0}^{n}$ satisfies the regularity and monotonity conditions in \citep{bretz2009graphical} for FWER control, which leads to 
		\begin{equation}\label{fwer-n}
		  \PP{A_n} \leq \sum_{i=1}^{n}\alpha_i \leq \alpha.
		\end{equation}
		Plugging into \eqnref{pfrecycling}, we have FWER upper bounded by $\alpha$ as claimed. The power superiority follows trivially from the dominance of each individual testing levels of Online Fallback over that of Alpha-Spending. Here we complete the proof for FWER control of Online Fallback.

		\section{Proof of \thmref{local}}\label{app:pflocal}
		Recall that $V$ is the number of false discoveries, which can be written as
		\[
		V = \sum_{i \in \nulls} R_i S_i , 
		\]
		by simply reiterating the obvious fact that $\alpha_i < \tau_i $, and hence for a discovery to be false (each term on the right hand side), it must be a null, it must be tested, and it must be rejected. Therefore,  
		\begin{equation}
			\EE{V} = \sum_{i\in\nulls} \EE{R_i S_i} = \sum_{i\in\nulls} \EE{\EEst{R_i S_i }{S_i,\F^{i-L_i-1}}} 
			= \sum_{i\in\nulls} \EE{\EEst{R_i}{S_i=1,\F^{i-L_i-1}}\PP{S_i=1}},
		\end{equation}
		where we used linearity of expectation twice. Here is the key step: since $\alpha_i, \lambda_i, \tau_i \in \F^{i-L-1}$, each term inside the last sum equals
		\begin{align}\label{eqlocal}
			\EE{\PPst{\frac{P_i}{\tau_i} \leq \frac{\alpha_i}{\tau_i}}{P_i \leq \tau_i,\F^{i-L_i-1}} \PP{S_i=1}} &\leq \EE{\frac{\alpha_i}{\tau_i} \PP{S_i=1}} \nonumber\\
			&\leq \EE{\frac{\alpha_i}{\tau_i} \EEst{\frac{\one{P_i\geq \lambda_i}}{1-\lambda_i/\tau_i}}{P_i\leq \tau_i, \F^{i-L_i - 1}}\PP{S_i=1}}\nonumber\\
			&= \EE{\frac{\alpha_i}{\tau_i} \EEst{\frac{\one{\lambda_i < P_i\leq \tau_i}}{1-\lambda_i/\tau_i}}{P_i\leq \tau_i, \F^{i-L_i - 1}}\PP{S_i=1}},
		\end{align}
		where we used the property of uniformly conservative nulls in \eqnref{conserve-def} twice.  Then using the fact that $\alpha_i < \tau_i$ are measurable with regard $\F^{i - L_i -1}$, and the law of iterated expectation, we have the last term in  \eqnref{eqlocal} equals to
		\begin{align*}
			&\EE{ \EEst{\alpha_i\frac{\one{\lambda_i<P_i\leq \tau_i}}{\tau_i-\lambda_i}}{S_i = 1, \F^{i-L_i - 1}}\PP{S_i=1}} = \EE{\EEst{\alpha_i \frac{\one{\lambda_i < P_i \leq \tau_i}}{\tau_i-\lambda_i}}{S_i, \F^{i-L_i - 1}}} 
			=   \EE{\alpha_i \frac{\one{\lambda_i < P_i \leq \tau_i}}{\tau_i-\lambda_i}}. 
		\end{align*}
		Therefore $\EE{V} \leq \sum_{i\in \nulls}\EE{\alpha_i \frac{\one{\lambda_i<P_i\leq \tau_i}}{\tau_i-\lambda_i}} =\EE{\sum_{i\in \nulls \cap \S \setminus \C} \frac{\alpha_i}{\tau_i-\lambda_i}}\leq \EE{\sum_{i\in \S \setminus \C} \frac{\alpha_i}{\tau_i-\lambda_i}} \leq \alpha$ by construction. 
	
		\section{Useful lemmas}
		\begin{lemma}\label{lem:lemmalim}
			For any $q > 1$, denote $\zeta(q)$ as the summation of sequence $\seq{i^{-q}}$, and $\zeta_l(q)$ as the summation of sequence $\seq{1/i\log^{q}{i}}$. Then for any constant $C>0$ and $\alpha \in (0,1)$, we have:
			\begin{itemize}		
				\item[(a)] $\lim_{x\to 0^{+}} x\exp{(-C\Phi^{-1}(x))}\log(1/x) = 0;$
				\item[(b)]  $\sum_{i=1}^{\infty} i^{-q}\exp{(-C\Phi^{-1}(\alpha i^{-q}/\zeta(q)))\log{i}} < \infty;$
				\item[(c)]  $\sum_{i=2}^{\infty} \frac{1}{i\log^q{i}}\exp{(-C\Phi^{-1}(\alpha \frac{1}{i\log^q{i}}/\zeta_l(q)))\log{(\log{i})}} = \infty;$
			\end{itemize}
			where $\Phi$ is the CDF of standard normal distribution.
			
		\end{lemma}
		
		\begin{proof}
		    We prove this lemma mainly use the following approximation of standard normal quantile and the dominate convergence theorem for series.  Firstly, we introduce the useful approximation, that is when $x \to 0^{+}$, 
			\begin{equation}\label{inverse}
			\Phi^{-1}(x) = -\sqrt{\log{\frac{1}{x^2}} - \log{(\log{\frac{1}{x^2}})} - \log{2\pi}} + o(1).
			\end{equation}
		    \paragraph{Proof for part (a)} For part (a), when $x\to 0^{+}$, using \eqnref{inverse}, we have
			\[
			x\exp{(-C\Phi^{-1}(x))}\log{\frac1x} = x \exp{\left\{C\left(\sqrt{\log{\frac{1}{x^2}} - \log{(\log{\frac{1}{x^2}})}-\log{2\pi}} - o(1)\right)\right\}} \log{\frac1x} \defn f(x)
			\]
			Taking log on $f(x)$, and assuming that $x< 1$, we obtain
			\begin{align}\label{eqa}
				\log{(f(x))} & = \log{x}  + C\sqrt{\log{\frac{1}{x^2}} - \log{(\log{\frac{1}{x^2}})}-\log{2\pi}} + \log{(\log{\frac{1}{x}})}- o(1) \nonumber\\ 
				& = -\log{\frac{1}{x}} + \log{(\log{\frac{1}{x}})} + C \sqrt{2\log{\frac{1}{x}} - \log{(\log{\frac{1}{x}})}- \log{4\pi}} - o(1)\nonumber\\ 
				& < -\log{\frac{1}{x}} + \log{(\log{\frac{1}{x}})} + C \sqrt{3\log{\frac{1}{x}} - 3\log{(\log{\frac{1}{x}})}-\log{4\pi}} - o(1),
			\end{align}
			where the last inequality is obtained using the fact that $y>\log{y^2}\ $ for all $y>0$. Let $z = \sqrt{\log{\frac{1}{x}} - \log{(\log{\frac{1}{x}})}-\frac13\log{4\pi}}$, then $x \to 0^{+}$ implies $z \to \infty$. Therefore, following \eqnref{eqa}, as $x \to 0^{+}$, we have
			\begin{equation}
			    \log{(f(x))} < - z^2 + \sqrt{3}C z - \frac13\log{4\pi} + o(1) \to -\infty,
			\end{equation}
			which implies that $f(x)$ goes to zero as $x$ goes to to zero, as we claimed in part (a). 
			
			\paragraph{Proof for part (b)} For part (b), we consider
			\begin{equation}
			f(i) \defn \frac{i^{q-1}\log^{-q-1}{i}}{\exp{(-C\Phi^{-1}(\alpha i^{-q}/\zeta(q)))}}.    
			\end{equation}
			Again using same trick in (a), we obtain
			\begin{align}\label{eqb}
				\log{(f(i))} & =(q-1)\log{i} -(q+1) \log{(\log{i})}-C\sqrt{\log{(i^q\zeta(q)/\alpha)^2} - \log{(\log{(i^q\zeta(q)/\alpha)^2})}-\log{2\pi}} -o(1) \nonumber\\
				& = (q-1)\log{i} - (q+1)\log{(\log{i})}
				-C\sqrt{2q\log{i} + 2\log{(\zeta(q)/\alpha)} - \log{(2q\log{ i} + 2\log{(\zeta(q)/\alpha)})} -\log{2\pi}} -o(1) \nonumber\\
				& > (q-1)\log{i} - (q+1)\log{(\log{i})} -C\sqrt{2q\log{i} - \log{(2q\log{ i} )}+ 2\log{(\zeta(q)/\alpha)} -\log{2\pi}  } -o(1) \nonumber\\
				& = (q-1)\log{i} - (q+1)\log{(\log{i})}-C\sqrt{2q\log{i}  - \log{(\log{i})} -A(q, \alpha)}-o(1),
			\end{align}
			where the inequality is obtained using the fact that $\log{(\zeta(q)/\alpha)}>0$, and we let $A(q, \alpha) \defn \log(4q\pi) - 2\log{(\zeta(q)/\alpha)}$ in the last step.
			Since $q>1$ and $\sqrt{\log(i)} = o(\log{i})$, therefore as $i$ goes to infinity, we have \eqnref{eqb} goes to infinity, which implies that  $\log{(f(i))}$ goes to infinity. This result gives us
			\begin{equation}
			\exp{(-C\Phi^{-1}(\alpha i^{-q}/\zeta(q)))} = o\left(i^{q-1}\log^{-q-1}{i}\right),
			\end{equation}
			which indicates
			\begin{equation}
			i^{-q}\exp{(-C\Phi^{-1}(\alpha i^{-q}/\zeta(q)))}\log{i} = o(\frac{1}{i\log^{q}{i}}).			    
			\end{equation}
			Therefore, for some $0<c<1$, there exist $N \in \N$, such that for any $i \geq N$, we have \[
			i^{-q}\exp{(-C\Phi^{-1}(\alpha i^{-q}/\zeta(q)))}\log{i} \leq  \frac{c}{i\log^{q}{i}}.
			\]
			Hence the summation 
			\[
			\sum_{i=1}^{\infty} i^{-q}\exp{(-C\Phi^{-1}(\alpha i^{-q}/\zeta(q)))\log{i}} \leq  	\sum_{i=1}^{N-1} i^{-q}\exp{(-C\Phi^{-1}(\alpha i^{-q}/\zeta(q)))\log{i}} + c \sum_{i=N}^{\infty} \frac{1}{i\log^{q}{i}} < \infty,
			\]
			which concludes part (b). 
			
			\paragraph{Proof for part (c)}
			In the end, for part (c), we use the similar technique in (a) and (b). Specifically, we consider
			\begin{equation}
			f(i) \defn \frac{\log^{q-1}{i}}{\exp{(-C\Phi^{-1}(\frac{\alpha}{i\log^{q}{i}}/\zeta_l(q)))}}.   
			\end{equation}
			Assume $i$ is big enough, then using approximation of $\Phi^{-1}$ in \eqnref{inverse}, and taking $\log$ on $f(i)$, we obtain
			\begin{align}\label{eqc}
				\log{(f(i))} & = (q-1)\log{(\log{i})}-C\sqrt{\log{(i\log^{q}{(i)}\zeta_l(q)/\alpha)^2} - \log{(\log{(i\log^{q}{(i)}\zeta_l(q)/\alpha)^2})}-\log{2\pi}} -o(1) \nonumber\\
				& \leq (q-1)\log{(\log{i})}-C\sqrt{\log{(i\log^q{(i)}\zeta_l(q)/\alpha)} -\log{2\pi}} - o(1)\nonumber\\
				& = (q-1)\log{(\log{i})}-C\sqrt{\log{i} + q\log{(\log{i})} + \log{(\zeta_l(q)/2\pi \alpha)} } - o(1)\nonumber\\
				&\to -\infty, \text{ as }\ \ i \to \infty
			\end{align}
			where the first inequality is obtained again using the simple fact that $y > \log{y^2}$ for all $y>0$, and last step uses the fact that $\log{(\log{i})} = o(\sqrt{\log{i}})$. Therefore, we have that $\log{(f(i))}$ goes to negative infinity, which implies $f(i)$ goes to zero. In other words, we obtain
			\[
			\log^{q-1}{i} =	o\left(\exp{(-C\Phi^{-1}(\frac{\alpha}{i\log^{q}{i}}/\zeta_l(q)))}\right),
			\]
			and therefore 
			\begin{equation}\label{eqc2}
			\frac{1}{i\log{i}} =	o\left(\frac{1}{i\log^q{i}}\exp{(-C\Phi^{-1}(\frac{\alpha}{i\log^{q}{i}}/\zeta_l(q)))}\right).		  
			\end{equation}
			Since $\{1/i\log{i}\}_{i=1}^{\infty}$ has infinite summation, it is easy to verify that the summation of terms in $o(\cdot)$ in \eqnref{eqc2} must be infinite too, which concludes part (c). 
		\end{proof}

		\begin{lemma}\label{lem:lemzeta}
			Denote $\zeta(q)=\sum_{i=1}^{\infty}i^{-q}$, and $\zeta^{\prime}$ be the first derivative of $\zeta$ with regard $q$, and assume $q>1$. Then for any fixed $p \geq 1$, we have that
			\begin{itemize}
			\item[(a)]  $0<\zeta(q) < \infty$, and is a smooth function that is monotonically decreasing with $q$;
			\item[(b)]$\zeta^{\prime}(q)  = \sum_{i=2}^{\infty} i^{-q}\log{i}$,  $-\infty<\zeta^{\prime}(q) < 0$, and is a smooth function that is monotonically increasing with $q$.
			\item[(c)]  $\lim_{q \to 1^{+}}\zeta(q) = -\lim_{q \to 1^{+}}\zeta^{\prime}(q)  = \infty; \ \ \ \lim_{q\to\infty}\zeta(q) = 1; \ \ \   \lim_{q\to\infty}\zeta^{\prime}(q) = 0.$
   	       \item[(d)]
			$\lim_{q\to\infty}\frac{\zeta^{\prime}(q)}{p\zeta(q)} = 0^{-}, \quad
			\lim_{q\to 1^{+}}\frac{\zeta^{\prime}(q)}{p\zeta(q)} = -\infty, \quad \lim_{q\to 1^{+}}\frac{\zeta^{\prime}(q)}{\zeta^2(q)} \in (-1, -1/2)$.
			\end{itemize}
		\end{lemma}
		
		\begin{proof}
		    (a-c) are established well-known results, therefore we omit the proofs and refer readers to resources like  \cite{devlin2002millennium} for more clues.  In the following, we will put our effort on proving (d), using (a-c) as known facts.
		    
			The first equation in (d) is obvious from (c). As for the second equation, given the fixed $p \geq 1$, denote $\psi_p(q) = \log{(\zeta^p(q))}$. From the properties of $\zeta(q)$ in (a-c), we have that $\psi_p(q)$ is also a smooth function with regard $q$, goes to infinity as $q$ goes to one, and is monotonically decreasing with $q$. Therefore, using contradiction, we must have the first derivative of $\psi_p(q)$ decreases to negative infinity as $q$ goes to one, that is
			\[
			\frac{\partial \psi_p(q)}{\partial q} = p \frac{\partial \log{(\zeta(q))}}{\partial q} = \frac{\zeta^{\prime}(q)}{p\zeta(q)} \to -\infty, \quad \text{as } q \to 1^{+}
			\]
			as we claimed.
			
			As for the third equation, denote $\varphi(q) =1/\zeta(q)$. Since $\zeta(q) > \int_{1}^{\infty} x^{-q} \diff x = \frac{1}{q-1}$ for $q > 1$, so $\varphi(q) < q-1$ for $q > 1$. Additionally noting that $\varphi(q) = 0 = q-1$ for $q=1$, we have 
			\[
			\frac{\partial \varphi(q)}{\partial q}\mid_{q=1} < \frac{\partial (q-1)}{\partial q}\mid_{q=1} = 1.
			\]
			On the other hand, $\zeta(q) < \int_{1}^{\infty} 2^q x^{-q} \diff x = \frac{2^q}{q-1}$ for any $q > 1$, so $\varphi(q) > \frac{q-1}{2^q}$ for any $q > 1$. Additionally noting that $\varphi(q) = 0 =\frac{q-1}{2^q}$ for $q=1$, we have 
			\[
			\frac{\partial \varphi(q)}{\partial q}\mid_{q=1} > \frac{\partial \frac{q-1}{2^q}}{\partial q}\mid_{q=1} = 1/2.
			\]
			Since $\frac{\partial \varphi(q)}{\partial q} = -\frac{\zeta^{\prime}(q)}{\zeta^2(q)}$, we have $\lim_{q\to 1^{+}}\frac{\zeta^{\prime}(q)}{\zeta^2(q)}\in (-1, -1/2)$ as claimed.
			
		\end{proof}

		\section{Proof for \thmref{fixed}}\label{app:fixedpf}
		We prove this theorem mainly using \lemref{lemmalim},  \lemref{lemzeta}, and the philosophy of contradiction, together with the results about Gaussian tail behaviour.
		
		\paragraph{Proof for part (a)}
		 Denote $\alpha_x(q) = \alpha x^{-q}/\zeta(q)$, where $x \in [1,\infty), q \in (1,\infty)$. From (a) in \lemref{lemzeta}, $\alpha_1(q)$ is monotonically increasing with $q$; for $x>1$,
        taking derivative of $\alpha_x(q)$ with regard $q$, we have
		\begin{equation}
		     \frac{\partial\alpha_x(q)}{\partial q}  = -\frac{\alpha x^{-q}}{\zeta^2(q)} (\zeta(q) \log{x}  + \zeta^{\prime}(q)),
		\end{equation}
		which shares the same sign with
		\[
		\beta_x(q) \defn -\log(x)-\zeta^{\prime}(q)/\zeta(q).
		\]
		It is easy to verify that $\beta_x(q)$ is monotonically decreasing with $q$, and using result (d) in \lemref{lemzeta}, we have $\lim_{q \to 1^{+}}\beta_x(q) = \infty$, and  $\lim_{q \to \infty}\beta_x(q) = -\log(x)$. Therefore, for $x>1$, $\beta_x(q)$ has unique zero point $q^{o}(x)$, which indicates  $\alpha_x(q)$ is monotonically increasing with $q$ for $q<q^{o}(x)$, and then monotonically decreasing with $q$ for $q>q^{o}(x)$. Additionally, note that $q^{o}(x)$ is decreasing with $x$, and goes to $1$ as $x$ goes to $\infty$.
		
		Then, noting that we can write $q$-series $\seq{\gamma_i}$ as $\gamma_i = \alpha_i(q)$ for all $i$, we have
		\begin{equation}\label{endprof}
		    \EEN{D} = \sum_{i = 1}^N \Phi(\Phi^{-1}(\alpha_x(q)) + \mu_A).    
		\end{equation}
		Since $\Phi$ and $\Phi^{-1}$ are both monotonically increasing function, therefore $\Phi(\Phi^{-1}(\alpha_i(q)) + \mu_A)$ shares the same fluctuation behaviour with $\alpha_i(q)$, that is, monotonically increasing with $q$ when $i=1$; increasing first then decreasing with $q$ when $i>1$.
		
		 we have $\sum_{i\geq2}^N \Phi(\Phi^{-1}(\alpha_i(q)) + \mu_A)$ firstly increases then decreases with $q$, since the derivative with regard $q$ for each term in the summation is a monotonic decreasing function with regard $q$ that starts from positive value and achieves negative value as $q$ goes from $1^{+}$ to $\infty$. Note that  $\Phi(\Phi^{-1}(\alpha_1(q)) + \mu_A)$ is monotonically increasing with $q$, therefore $\EEN{D}$ will be either just monotonically increasing with $q$, or increasing first and then decreases with $q$. In order to settle this, it is suffice to look at the derivative of $\mathbb{E}_2[D]$ with regard $q$, and check whether it achieves negative value at some point. Specifically, we have
		\begin{align}\label{en}
			\frac{\partial\mathbb{E}_2[D]}{\partial q}  &= -C(\pi_A, \mu_A)\sum_{i=1}^{2} \alpha_i(q)\exp{\left(-\mu_A \Phi^{-1}(\alpha_i(q))\right)}\left(\log{i} + \zeta^{\prime}(q)/\zeta(q)\right)\defn e_2(q),
		\end{align}
		\noindent
		where $C(\pi_A, \mu_A) =  \pi_A \exp{(-\mu_A^2/2)}>0$. 
 		Note that $e_2(q)$ shares the same sign with
		\begin{align}\label{neg}
          & - \alpha_1(q)\exp{\left(-\mu_A \Phi^{-1}(\alpha_1(q))\right)} - \alpha \log{2}\alpha_2(q)\exp{\left(-\mu_A \Phi^{-1}(\alpha_2(q))\right)} (\zeta(q)/\zeta^{\prime}(q) + 1/\log(2))\nonumber\\
          &\defn -C_1(q) - C_2 \frac{2^{-q}\log{2}}{\zeta^{\prime}(q)}A_2(q) - C_3\alpha_2(q)A_2(q),
 		\end{align}
 		where $C_1(q) \defn \alpha_1(q)\exp{\left(-\mu_A \Phi^{-1}(\alpha_1(q))\right)}$, $C_2 \defn \alpha$, $C_3 \defn \alpha \log{2}$, $A_2(q) = \exp{\left(-\mu_A \Phi^{-1}(\alpha_2(q))\right)}$. 
 		 
 		Then we analyze each term in \eqref{neg} one by one, as $q\to \infty$. Using (a) in \lemref{lemzeta}, we have that, $\lim_{q\to\infty}C_1(q)= \alpha\exp{\left(-\mu_A \Phi^{-1}(\alpha)\right)}$, which is a positive constant; using (a) in \lemref{lemmalim}, and the fact that $1/\alpha_2(q) > 1$, and $\lim_{q\to \infty}\alpha_2(q)\to 0$, we have $\lim_{q\to \infty}\alpha_2(q) A_2(q) \to 0$; and $\lim_{q\to \infty} A_2(q) \to \infty$. Finally, we look at the term $\frac{2^{-q}\log{2}}{\zeta^{\prime}(q)}$, and we prove that it goes to one as $q$ goes to infinity. Using (b) in \lemref{lemzeta}, we have 
 		\begin{align}
 		    \frac{2^{-q}\log{2}}{\zeta^{\prime}(q)} &= \frac{2^{-q}\log{2}}{\sum_{i=2}i^{-q}\log{i}} = \frac{1}{1+ \frac{1}{\log{2}}\sum_{i=3}{(\frac{i}{2})}^{-q}\log{i}} > \frac{1}{1+ \frac{1}{\log{2}} \int_{x=2}^{\infty} (\frac{x}{2})^{-q} \log{x}} \nonumber\\
 		    &= \frac{1}{1+\frac{1}{\log{2}}  (\frac{1}{1-q}-\frac12)} \to \frac{1}{1-\frac{1}{2\log{2}}}\quad  \textnormal{ as } q \to \infty.  
 		\end{align}
 		Combining the above analysis, we have \eqref{neg} goes to $-\infty$ as $q\to\infty$, and therefore $\lim_{q\to\infty}e_2(q)<0$.

		Therefore, $\mathbb{E}_2[D]$ and hence $\EEN{D}$ must increases with $q$ first and then decreases with $q$ as we claimed, and there must exist an optimal $q^{*}$ which depend on $\mu_A$ and $N$, such that $\EEN{D}$ achieve its maximum. In fact,
		\begin{equation*}
			q^{*}(\mu_A, N) = \textnormal{argzero} \{e_N(q)\}, 
		\end{equation*}
		where \textnormal{argzero}$\{\cdot\}$ represents the zero point of a function.

		In the following, we prove that $q^{*}(\mu_A, N)$ goes to $1^{+}$ as $N$ goes to infinity. Specifically, we proceed this using contradiction. Since $\mu_A$ is fixed constant in this context, we write $q^{*}(\mu_A, N)$ as $q^{*}_N$ for simplicity from now on. Firstly, for any $N \in \N$, we must have
		\begin{equation}
		    \log{N} + \zeta^{\prime}(q^{*}_N)/\zeta(q^{*}_N) > 0,
		\end{equation}
		otherwise $e_N(q^{*}_N)$ will be negative which contradicts the the definition of $q^{*}_N$. Then, for $N_1 > N$, we have
		\begin{align}
		e_{N_1}(q^{*}_N) 
		&= e_{N}(q^{*}_N) -C \sum_{i = N+1}^{N_1}\alpha_i(q^{*}_N)\exp{(-\mu_A \Phi^{-1}(\alpha_i(q^{*}_N)))}(\log{i} + \zeta^{\prime}(q^{*}_N)/\zeta(q^{*}_N)) \nonumber\\
	    &\leq 0 -C \sum_{i = N+1}^{N_1}\alpha_i(q^{*}_N)\exp{(-\mu_A \Phi^{-1}(\alpha_i(q^{*}_N)))}(\log{N} + \zeta^{\prime}(q^{*}_N)/\zeta(q^{*}_N)) < 0.
		\end{align}
        Since $e_{N}$ must be decreasing with $q$ near the zero point for any $N$, therefore we must have $q^{*}_N < q^{*}_{N_1}$, which means $q^{*}_N$ is an monotonically decreasing function of $N$.
		
		We then use contradiction to prove the argument about the limit. If $\lim_{N\to \infty} q^{*}_N \neq 1$, without lose of generality, let's say the limit is some constant $c > 1$. Then, since $q^{*}_N$ is monotonically decreasing with $N$, so we must have that $q^{*}_N > c$ for any $N \in \mathbb{N}$, which means $e_{N}(c) > 0$ for any $N \in \mathbb{N}$, that is
		\begin{equation}
		\sum_{i=1}^{\infty} i^{-c}\exp{(-\mu_A\Phi^{-1}(\alpha i^{-c}/\zeta(c)))}(\log{i} + \zeta^{\prime}(c)/\zeta(c)) > 0.		    
		\end{equation}
		From \lemref{lemmalim}(b), we know that the summation is finite and hence the above inequality implies that
		\begin{equation}
		\sum_{i=1}^{\infty} i^{-c}\exp{(-\mu_A\Phi^{-1}(\alpha i^{-c}/\zeta(c)))}\log{i} > -\zeta^{\prime}(c)/\zeta(c)\sum_{i=1}^{\infty} i^{-c}\exp{(-\mu_A\Phi^{-1}(\alpha i^{-c}/\zeta(c)))},
		\end{equation}
		which can be rearranged to the claim that
		\begin{equation}\label{wrong}
		\frac{\sum_{i=1}^{\infty} i^{-c}\exp{(-\mu_A\Phi^{-1}(\alpha i^{-c}/\zeta(c)))}}{\sum_{i=2}^{\infty} i^{-c}\exp{(-\mu_A\Phi^{-1}(\alpha i^{-c}/\zeta(c)))}\log{i}} >  \frac{\sum_{i=1}^{\infty}i^{-c}}{\sum_{i=2}^{\infty}i^{-c}\log{i}}.		    
		\end{equation}
		However, for any increasing sequence $\seq{c_i}$ with $c_i > 1$ and $\lim_{i \to \infty}c_i = \infty$, and any sequences $\seq{b_i}$, and $\seq{a_i}$, such that $b_i > a_i > 0$, we in fact have
		\begin{equation}
		\frac{\sum_{i=1}^\infty c_i a_i}{\sum_{i=2}^\infty c_i b_i} < \frac{\sum_{i=1}^\infty a_i}{\sum_{i=2}^\infty b_i}.		    
		\end{equation}
		Substituting $a_i=i^{-c}, b_i = i^{-c} \log i$ and $c_i=\exp{(-\mu_A\Phi^{-1}(\alpha i^{-c}/\zeta(c)))}$ into the above fact, we conclude that inequality \eqref{wrong} cannot be true. In conclusion, such a limit $c>1$ does not exist. Therefore, the limit must equal $1$, completing the proof for part (a).
		
		\paragraph{Proof for part (b)}
		
		Denote $t(x)\defn \Phi^{-1}(C x^{-q})$, where $C \equiv \alpha /\zeta(q)$ is some positive constant depend on $q$ and $\alpha$, with $\zeta(q)$ defined in part (a). Taking second derivative of $t$ with regard $x$, we have
		\begin{equation}\label{second}
			\frac{\partial^2 t(x)}{\partial x^2} = \frac{\partial^2 \Phi^{-1}(Cx^{-q})}{\partial x^2} = Cq\frac{(q+1)\phi(\Phi^{-1}(Cx^{-q}))x^q-Cq\Phi^{-1}(Cx^{-q})}{\phi^2(\Phi^{-1}(x^{-q}))x^{2(q+1)}}.
		\end{equation}
		Note that the expression \eqref{second} shares the same sign with $ (q+1)\phi(-t)+q t\overline{\Phi}(-t)$, 
		which is positive when $t < 0$, from the fact that $\overline{\Phi}(y) < \phi(y)/y$ for all $y > 0$. Therefore, we have that, the second derivative is positive when $t(x)<0$, i.e. $x > \sqrt[q]{2C}$. In other words, denote the first derivative of $t$ with regard $x$ as $t^{\prime}$, we have that $t^{\prime}$ is an monotonically increasing function of $x$, for $x > \sqrt[q]{2C}$. 
		
	    On the other hand, for $q$-series $\seq{\gamma_i}$, the testing levels of Alpha-Spending are $\alpha_i \defn \alpha\gamma_i = C i^{-q}$ for all $i \in \N$, and 
		\begin{equation}
		    \Phi(\Phi^{-1}(\alpha_i)+\mu_A) \to 0, \quad \text{as } i \to \infty.
		\end{equation}
		Therefore, for all $N \in \N$, there exists $K(N)> N$, such that
		\begin{equation}\label{kn}
		    \Phi(\Phi^{-1}(\alpha_n)+\mu_A) \leq \alpha_N, \quad \text{ for all } n \geq K(N).
		\end{equation}
		Choose some $N > \sqrt[q]{2C}$, we have
		\begin{equation}\label{tn}
		t(N) - t(N+1) =  t^{\prime}(\epsilon_1) \leq t^{\prime}(\epsilon_2) = t(K(N)) - t(K(N)+1)
		\end{equation}
		holds true for some $\epsilon_1 \in (N, N+1)$, $\epsilon_2 \in (K(N), K(N)+1)$, and $ K(N) > N$, due to the monotonicity of $t^{\prime}$ as we described before. Plugging in the expression for $\alpha_n$, result \eqref{tn} is equivalent to
		\[
		\Phi^{-1}(\alpha_N)-\Phi^{-1}(\alpha_{N+1}) \leq  \Phi^{-1}(\alpha_{K(N)})-\Phi^{-1}(\alpha_{K(N)+1}).
		\]
		Rearranging the terms and using result \eqref{kn}, we have
		\[
		\Phi^{-1}(\alpha_{K(N)+1}) + \mu_A < \Phi^{-1}(\alpha_{K(N)}) + \mu_A - \Phi^{-1}(\alpha_N) + \Phi^{-1}(\alpha_{N+1}) \leq \Phi^{-1}(\alpha_{N+1}).
		\]
		Since $\Phi$ is a monotonically increasing onto mapping, we in turn have
		\[
		\Phi(\Phi^{-1}(\alpha_{K(N)+1}) + \mu_A) \leq \Phi(\Phi^{-1}(\alpha_{N+1})) = \alpha_{N+1}.
		\]
		Iteratively using the same tricks, we can prove that
		\[
		\Phi(\Phi^{-1}(\alpha_{K(N)+i}) + \mu_A) \leq \Phi(\Phi^{-1}(\alpha_{N+i})) = \alpha_{N+i}\ \ \text{ for all } i \in \mathbb{N}.
		\]
		Therefore, $\EE{D}\leq \sum_{j=0}^{K(N)}\pi_A\Phi(\Phi^{-1}(\alpha_{j})+\mu_A) + \sum_{i=1}^{\infty}\pi_A\alpha_{N+i} \leq \pi_A(C + \alpha)< \infty$ as we claimed. Also, as indicated by part (a), $\EE{D}$ is monotonically decreasing with $q$. Hence we finished the proof for part (b).

    \section{Optimal configuration for Alpha-Spending with varying signal strength and density}\label{app:varied}
	
	When we are provided some prior information about the hypothesis, we should be allowed to achieve optimal power by choosing some more ambitious (i.e. decays faster) sequence $\seq{\gamma_i}$. Specifically, we consider two forms of prior knowledge: (a) we know information about which hypotheses are more likely to be nulls; (b) we know information about the signal strength of different non-nulls. These two forms of prior knowledge just correspond to the setting that $\seq{\pi_{Ai}}$ and $\seq{\mu_{Ai}}$ are varied constant sequences. In this setting, we prove that when the sequences $\seq{\pi_{Ai}}$ and $\seq{\mu_{Ai}}$ are nice (i.e. satisfies some reasonable conditions), there exists an optimal sequence $\seq{\gamma_i}$ for maximizing the  power, where each term $\gamma_i$ is function of $\pi_{Ai}$ and $\mu_{Ai}$. We state the specific results as below:
	\begin{theorem}\label{thm:alpha}
		For $\seq{\pi_{Ai}}$ and $\seq{\mu_{Ai}}$, where for all $i\in \N$, $\pi_{Ai}\in (0,1)$ and $\mu_{Ai} > \epsilon$ for some $\epsilon > 0$, we define
		\begin{equation}\label{hidef}
		h_i(\eta) = \frac{1}{\mu_{Ai}}\log{\frac{\eta}{\pi_{Ai}}} + \mu_{Ai}, \quad \text{where } \eta \in (0,\infty). 
		\end{equation}
		If there exist a sequence $\seq{1/b_i}$ with finite summation, such that
		\begin{equation}\label{cond1}
			\frac{\sqrt{\log{b_i}}}{h_i(1)} = o(1),
		\end{equation}
		then there exists some $\eta^{*} > 0$, such that sequence $\seq{\gamma_i^{*}}$ where
		\begin{equation}\label{exp-a}
			\gamma_i^{*} = g_i(\eta^{*}) \defn  \frac{1}{\alpha}\Phi(-h_i(\eta^{*}))
		\end{equation}
		sums to one and achieves the highest power of Alpha-Spending. 
	\end{theorem}
	\noindent
	 The condition \eqref{cond1} is in fact easy to satisfy. For example, let $\pi_{Ai} = i^{-q}$ and $\mu_{Ai} = \log{i^{q}}$, where $q > 1$, then we have condition \eqref{cond1} holds true for any $b_i \propto i^{r}$, where $r<q$. \thmref{alpha} in fact tells us that, as long as the sequences $\seq{\pi_{Ai}}$ and $\seq{\mu_{Ai}}$ satisfy some reasonable conditions, we can find a closed form of optimal sequence $\seq{\gamma_i^{*}}$ for power function, and this optimal sequence $\seq{\gamma_i^{*}}$ does not need to be monotonically decreasing. Another interesting fact implied by \thmref{alpha} is that, if $\pi_{Ai}$ is fixed as $\pi_A$, then the optimal $\gamma_i^{*}$ remains the same for $\pi_A$ of different values, which is easy to verify from the form of $\gamma_i^{*}$. 
	
      \begin{proof}
		We would like to obtain the optimal sequence $\seq{\gamma_i^{*}}$ to maximize $\EE{D}$, under the constraint of $\sum_{i=1}^{\infty}\gamma_i^{*} = 1$. Naturally, we resort to Lagrange Multiplier for solutions. Specifically, the corresponding Lagrange multiplier is:
		\[
		\mathcal{L} = \sum_{i=0}^{\infty}\pi_{Ai}\Phi(\Phi^{-1}(\alpha\gamma_i) + \mu_{Ai}) - r(\sum_{i=1}^{\infty}\gamma_i^{*} -1).
		\]
		Taking derivative with regard each $\gamma_{i}$ and set them to zeros, where $i = 1, 2, \dots$, we have:
		\begin{equation}\label{bestgamma}
		\gamma_i^{*} =\frac{1}{\alpha} \Phi(\frac{1}{\mu_{Ai}}\log{(\alpha\pi_{Ai}/r)} -  \frac{\mu_{Ai}}{2}) = \frac{1}{\alpha}\Phi(-h_i(\eta^{*}))\defn y_i(\eta^{*})		    
		\end{equation}
		where $\eta^{*} \defn r/\alpha$ is chosen to satisfy the constraint $\sum_{i=1}^{\infty}\gamma_i^{*} = 1$, and $h_i(\eta)\defn\frac{1}{\mu_{Ai}}\log{(\pi_{Ai}/\eta)} -  \frac{\mu_{Ai}}{2}$. Next, we prove that condition \eqref{cond1} guarantees that such valid $\eta^{*}$ exists, by proving that given condition \eqref{cond1}, there will always exists some $\eta^{*} >0$, such that $\seq{\gamma_i^{*}}$ defined in \eqnref{bestgamma} have summation equals to one. 
		
		First we want to prove that condition \eqref{cond1} implies $\log{b_i}/|h_i(\eta)^2| = o(1)$ for all $\eta > 0$. Note that 
	    \begin{equation}\label{hi}
	        h_i(\eta) = \frac{1}{\mu_{Ai}}\log{(1/\pi_{Ai})} +\frac{1}{\mu_{Ai}}\log{\eta}+ \frac{\mu_{Ai}}{2} = h_i(1)\left( 1 + \frac{\log{\eta}}{\log{1/\pi_{Ai}} + \frac{\mu_{Ai}^2}{2}}\right) \defn  h_i(1) (1 + g_i(\eta)).
	    \end{equation}
	    The argument is obvious for $\eta>1$, since $\mu_{Ai} >0$ for any $i$, therefore $\log{b_i}/|h_i(\eta)^2| < \log{b_i}/|h_i(1)^2|$.
		
		For the case $0<\eta<1$, we consider the following. Condition \eqref{cond1} implies that
		\begin{equation}\label{hilim}
		 \lim_{i\to \infty}h_i(1) = \infty.  
		\end{equation}
		From the expression of $h_i(\eta)$ in \eqnref{hi}, result \eqref{hilim} indicates at least one of the following is true:
		\begin{equation}\label{twolim}
		\lim_{i \to \infty}\mu_{Ai} = \infty, \quad \lim_{i \to \infty}\log{(1/\pi_{Ai})}/\mu_{Ai} = \infty.
		\end{equation}
		If $\mu_{Ai}$ goes to infinity, then obviously $g_i(\eta)$ in \eqnref{hi} goes to zero; If $\mu_{Ai}$ is upper bounded, then $\pi_{Ai}$ must go to zero from \eqnref{twolim}, which leads to $g_i(\eta)$ goes to zero as well; If $\seq{\mu_{Ai}}$ is neither of the above two cases, then it must has a subsequence $\{\mu_{Ai_k}\}_{k = 1}^{\infty}$ that goes to infinite, and is finite for the rest terms, since we require it to be lower bounded by some positive constant. Therefore,  $\pi_{Ai}$ goes to zero for $i \neq i_k$, and thus $g_i(\eta)$ goes to zero for $i \neq i_k$. At the same time we also have $g_{i_k}(\eta)$ goes to zero, since $\mu_{Ai_k}$ goes to infinity. Therefore, we have $g_i(\eta)$ always goes to zero for any $\eta>0$, which indicates $h_i(\eta) = O( h_i(1))$ for any $\eta>0$, and consequently
		\begin{equation}\label{hieta}
		\log{b_i}/|h_i(\eta)^2| = o(1),\quad \text{ for all } \eta>0.		    
		\end{equation}
		Equation \eqref{hieta} indicates there exists $N \in \N$, such that for any $i > N$ we have that
		\[
		\Phi(-h_i(\eta))  \leq c\exp{(-h_i(1)^2/2)} \leq c\exp{(-\log{b_i})} = \frac{c}{b_i},
		\]
		where $c$ is some positive constant. Therefore, we have that, for all $\eta>0$,
		\[
		\sum_{i=1}^{\infty} y_i (\eta)\defn \frac{1}{\alpha} \sum_{i=1}^{\infty} \Phi(-h_i(\eta)) \leq  \frac{1}{\alpha}(\sum_{i=1}^{N} \Phi(-h_i(\eta))+ c\sum_{i>N}\frac{1}{b_i}) < \infty.
		\] 
		Since $h_i$ is an continuous increasing function with regard $\eta$ for all $i$, we have that $\sum_{i=1}^{\infty}  y_i$ is a continuous decreasing function with regard $\eta$.  Additionally, note that $h_i(0) = -\infty$, and $h_i(\infty) = \infty$, which implies $ y_i(0) = 1/\alpha$ and $y_i(\infty) = 0$.  Therefore, $\sum_{i=1}^{\infty}y_i $ is a one to one mapping from $(0,\infty)$ to $(0, \infty)$, which indicates that there must exist some $\eta^{*} > 0$, such that $\sum_{i=1}^{\infty}y_i(\eta^{*})  = 1 \in (0, \infty)$. Hence we finished proving that,  under condition \eqref{cond1}, there always exists a valid $\eta^{*}$ such that the optimal seqeunce $\seq{\gamma_i^{*}}$ where $\gamma_i^{*} = y_i(\eta^{*})$ for all $i$ sums to one. 
		\end{proof}

      \section{Proof for \thmref{more}} \label{app:pfmore}
		We prove this theorem mainly using the weak law of large numbers, which in turn allow us to deduce sufficient conditions with regard parameters $\tau$ and $\lambda$ for the arguments being true. Then we derive the closed form of these sufficient conditions, and comment on their properties.  
		
		\paragraph{Proof for part (a)} Firstly, in order to prove argument (a), we denote the individual testing level for the $i$-th $H_i$ given by Alpha-Spending as $\alpha_i$, and as $\widetilde{\alpha}_i$ given by our new method--- ADDIS-Spending. Specifically, we have
		\begin{equation*}
		   \alpha_i = \alpha \gamma_i, \quad \widetilde{\alpha}_i = \alpha (\tau-\lambda)  \gamma_{1 + \sum_{j < i} S_j-C_j},
		\end{equation*}
		where $\seq{\gamma_i}$ is some nonnegative sequence that sums to one. We particularly assume it is a log-$q$-series as defined in \defref{q-series}. Then, denoting $S^t = \sum_{j \leq t} S_j$ and denoting $C^t = \sum_{j \leq t} C_j$, using the weak law of large numbers, we have
		\begin{align}\label{probconv}
     	\frac{\widetilde{\alpha}_t}{\alpha_t} &= \frac{(\tau-\lambda) \gamma_{1+S^{t-1} - C^{t-1}}}{\gamma_t} \nonumber\\
     	&=(\tau-\lambda) \frac{t\log^q{t}}{(1+S^{t-1}-C^{t-1})\log^{q}{(1+S^{t-1}-C^{t-1})}}\nonumber\\
			& \geq  (\tau-\lambda) \frac{t}{1 + S^{t-1}-C^{t-1}} = (\tau-\lambda)  \left( \frac{1}{t}(1+\sum_{i< t}\one{ \lambda< P_i \leq \tau}) \right)^{-1} \nonumber\\
			&\xrightarrow{p} (\tau-\lambda) \left( \PP{ \lambda < P_i \leq \tau}\right)^{-1} \defn C(\lambda, \tau),
		\end{align}
		where $\xrightarrow{p}$ means converging in probability. In other words, \eqref{probconv} implies that, for any $\epsilon > 0$, $\delta > 0$, there exists $N(\delta, \epsilon) \in \mathbb{N}$, such that for any $t>N$, we have
		\[
		\PP{\left|\frac{t (\tau-\lambda) }{1 + S^{t-1} -  C^{t-1}} - C(\lambda, \tau)\right| > \epsilon} \leq \delta.
		\]
		If $C(\lambda, \tau)>1$, then there exists $\epsilon^{*}>0$, such that $C(\lambda, \tau) \geq 1+2\epsilon^{*}$. Given such $\epsilon^{*}$, and any $\delta > 0$, we have that for any $t> N(\epsilon^{*}, \delta)$, with probability at least 1- $\delta$,
		\[\frac{\widetilde{\alpha}_t}{\alpha_t} \geq \frac{t((\tau-\lambda))}{1 + S^{t-1}-C^{t-1}} \geq C(\lambda, \tau)-\epsilon^{*} \geq 1+\epsilon^{*}.\] 
		This result in turn gives us that, with probability at least $1-\delta$, for any $n > N$, we have
		\begin{equation}\label{disproof}
			\begin{split}
				&\ \ \ \ \  \EE{D_{\textnormal{ADDIS}}(\lambda, \tau)} - \EE{D_{\textnormal{spend}}}\\ 
				& = \pi_A \left[ \sum_{i=2}^{\infty} \Phi(\Phi^{-1}(\widetilde{\alpha}_i)+ \mu_A) -  \Phi(\Phi^{-1}(\alpha_i)+ \mu_A) \right] \\
				&\geq \pi_A \left[ \sum_{2\leq i\leq N} \Phi(\Phi^{-1}((\tau-\lambda)\alpha_i)+ \mu_A) -  \Phi(\Phi^{-1}(\alpha_i )+ \mu_A) + \sum_{i > N} \Phi(\Phi^{-1}((1+\epsilon^{*})\alpha_i ) + \mu_A) -  \Phi(\Phi^{-1}(\alpha_i) + \mu_A) \right]\\
				& \geq \pi_A \exp{(-\mu_A^2/2)} \left[\left((\tau-\lambda)- 1\right)\sum_{2\leq i\leq N} \exp{\left(-\mu_A\Phi^{-1}((\tau-\lambda)\alpha_i)\right)}\alpha_i  + \epsilon^{*}\sum_{i > N} \exp{\left(-\mu_A\Phi^{-1}((1+\epsilon^{*})\alpha_i)\right)}\alpha_i \right],
			\end{split}
		\end{equation}
		where the first inequality is obtained using the fact that $\Phi(\Phi^{-1}(x) + \mu_A)$ is an monotonically increasing function with regard $x$, and that $\widetilde{\alpha}_i> (\tau-\lambda)\alpha_i$ by construction, while $\widetilde{\alpha}_i\geq(1+\epsilon)\alpha_i$ for $i > N$. The second inequality comes from the smoothness of function $\Phi(\Phi^{-1}(x) + \mu_A)$ with regard $x$, and the fact that its derivative with regard $x$ is also a monotonically increasing function.
		
		From part(c) in \lemref{lemmalim}, we know that the second summation in \eqref{disproof} is unbounded, while the first summation is a finite constant, therefore expression  \eqref{disproof} is bigger than zero. In conclusion, as long as  $C(\lambda, \tau)>1$, we will have $\EE{D_{\textnormal{ADDIS}}(\lambda, \tau)} \geq \EE{D}$ with probability at least $1-\delta$ for any $\delta > 0$, which is saying,  $\EE{D_{\textnormal{ADDIS}}(\lambda, \tau)} \geq \EE{D}$ with probability one if $C(\lambda, \tau)>1$.\\
		
		\noindent
		In the following, we derive the closed form for $\tau$ to satisfy the sufficient condition for our claim--- $C(\lambda, \tau)>1$. Denote $\mathcal{I}\defn \{\tau \in [0,1]:  C(\lambda,\tau)>1\}$. Note that
		\begin{equation}\label{Clt}
		    C(\lambda,\tau)>1\ \Leftrightarrow\ (\tau-\lambda) (G(\tau) - G(\lambda))^{-1} > 1\ \Leftarrow\ \tau \geq G(\tau),\ \lambda \leq G(\lambda),\ 0 < \tau-\lambda < 1.
		\end{equation}
	    Denote $J(x) \defn x - G(x)$,
		and $G(x) = (1-\pi_A)\Phi(\Phi^{-1}(x)+\mu_N) + \pi_A \Phi(\Phi^{-1}(x)+\mu_A)$, which is the joint CDF of both nulls and non-nulls, where $\Phi$ is the CDF of standard Gaussian as always. 
		
		As an first obvious fact, $J(0) = J(1) = 0$. Then, note that  $G^{\prime}(x) = C_1 \exp{(c_1 \Phi^{-1}(x))} + C_2 \exp{(-c_2\Phi^{-1}(x))}$, where $C_1, C_2, c_1, c_2 $ are some positive constants depend on $\mu_A, \mu_N, \pi_A$. Letting $y = \exp{(\Phi^{-1}(x))}$, we have
		\begin{equation}
		J^{\prime}(x)  = h(y) \defn 1 - C_1 y^{c_1} - C_2 \frac{1}{y^{c_2}}.
		\end{equation} 
		Since $y > 0$, so $h(y)$ is increasing with $y$ first and then decreasing with $y$, that is  $J^{\prime}(x)$ is increasing with $x$ first and then decreasing with $x$. On the other hand, particularly, using the tail behaviour of Gaussian, it is also easy to verify that  
		\begin{equation}
		\lim_{x \to 0} J^{\prime}(x) =  \lim_{x \to 1} J^{\prime}(x) =  -\infty.
		\end{equation}
		Based on above discussion, we can conclude that, $J$ must be a continuous function that is zero at zero, decreasing with $x$ first, and then increasing to some positive constant, and then decreasing to zero at the endpoint one. Denote the only zero point of $J$ that is not endpoint as $c^{*}$, it is now obvious that 
		\[J(x)\geq 0 \Leftrightarrow x \in \{0\}\cup [c^{*}, 1], \textnormal{ and } J(x) \leq 0 \Leftrightarrow x \in [0, c^{*}] \cup\{1\}.\]
		 
		Therefore, the sufficient condition for $C(\lambda, \tau) >1$ in \eqref{Clt} is equivalent to 
		\[
		\lambda \in [0, c^{\star}),\quad \tau \in (c^{\star}, 1],\quad \tau-\lambda < 1,
		\]
		which concludes our argument.

		By construction, the lower bound $c^{*}$ depends on parameters $\pi_A, \mu_N$ and $\mu_A$. From the following intuitive interpretation of $J$, we will see that it is in fact increasing with $\pi_A, \mu_N$ and $\mu_A$. Since $\Phi(\Phi^{-1}(x) + \mu_N) < x$, $\Phi(\Phi^{-1}(x) + \mu_A) > x$, therefore, $\pi_A$ represents the weight of the part that makes $J(x)$ smaller, which indicates that $J(x)$ should decreases with $\pi_A$. Consequently, as $\pi_A$ increases, the zero point $c^{*}$ will increase, since $J$ is increasing near the zero point. Similarly, one may argue that $c^{*}$ is an increasing function of $\mu_A$ and $\mu_N$. An additional interesting fact that is easy to derive using the above steam of analysis is that, as $\mu_N$ achieves exactly zero, that is the nulls are exactly uniform, then $c^{*}$ is exactly one, for any $\pi_A \in (0,1)$, and $\mu_A >0$. 
		
		


	\section{Application of adaptivity and discarding}\label{app:addisapp}
	In fact, the ideas of adaptivity and discarding in \secref{ADDIS-Spending} can also be applied to methods in \secref{current} to develop more powerful variants. Recall the definitions in and right after  \eqnref{indicators}, we present the following examples.
	\begin{itemize}
		\item{\textbf{Discard-Sidak}} For FWER level $\alpha$, let the fileration $\F^i = \sigma(S_{1:i}, R_{1:i})$ for some predictable discarding level $\seq{\tau_i} \in (0,1)$ where $\tau_i \geq \alpha$ for all $i$. Discard-Sidak then tests each hypothesis $H_i$ at level $\alpha_i$, where
		\begin{equation}
			\alpha_i = \tau_i(1 - (1-\alpha)^{\beta_i}), \quad \text{with } \sum_{j \in \S} \beta_j \leq 1.
		\end{equation}
		\item{\textbf{Adaptive-Sidak}} For FWER level $\alpha$, let the fileration $\F^i = \sigma\{C_{1:i}, R_{1:i}\}$ for some predictable candidate level $\seq{\lambda_i} \in (0,1)$. Adaptive-Sidak then tests each hypothesis $H_i$ at level $\alpha_i$, where 
		\begin{equation}
			\alpha_i = 1 - (1-\alpha)^{\beta_i}, \quad \text{with } \sum_{j \notin \C} \beta_j/(1-\lambda_i) \leq 1.
		\end{equation}

		\item{\textbf{ADDIS-Sidak}} For FWER level $\alpha$, let the fileration $\F^{i} = \sigma(S_{1:i}, C_{1:i}, R_{1:i})$ for some predictable candidate level $\seq{\lambda_i}$ and discarding level $\seq{\tau_i}$ where we require $\max\{\lambda_i, \alpha\} \leq \tau_i  \in (0,1)$  for all $i$. ADDIS-Sidak then tests each hypothesis $H_i$ at level $\alpha_i$, where 
		\begin{equation}
			\alpha_i =\tau_i(1 - (1-\alpha)^{\beta_i}), \quad \text{with }\sum_{j \in \S \setminus{\C}} \beta_j/(1-\lambda_i) \leq 1.
		\end{equation}
		
		\item{\textbf{Discard-Fallback}}
		For FWER level $\alpha$, let the fileration $\F^i = \sigma(S_{1:i}, R_{1:i})$ for some predictable discarding level $\seq{\tau_i}$ where each term $\tau_i \geq \alpha$. Then Discard-Fallback tests $H_i$ at level 
		\begin{equation}
			\alpha_i := \tau_i (\alpha \gamma_{1+\sum_{j<i}S_i} + \sum_{k=1}^{i-1}w_{\delta_k, i}R_{\delta_k} \alpha_{\delta_k}),
		\end{equation} 
		where $\delta_k$ is the index of the last hypothesis that is not discarded before $k+1$, and $\{w_{ki}\}_{i=k+1}^{\infty}$ is some infinite sequence that is nonnegative and sums to one for all $k \in \N$.
	\end{itemize}
	
	\noindent 
	Like methods in \secref{ADDIS-Spending}, all the variants mentioned above have FWER control when all null $p$-values are independent of each other and of non-nulls. Specifically, we present the following \propref{thmextend}.
	\begin{proposition} \label{prop:thmextend}
		When all null $p$-values are independent of each other and of non-nulls, Adaptive-Sidak controls the FWER.
		Additionally, if the nulls are uniformly conservative \eqref{conserve-def}, then we have FWER control for Discard-Sidak, ADDIS-Sidak and Discard-Fallback.
	\end{proposition}
	\begin{proof}
		We prove this Proposition using the similar techniques in the proof of theorems in \secref{ADDIS-Spending}, with the following \lemref{pos-asso} that comes from the positive association property of independent random variable introduced in \citep{esary1967association}. 
		
		\begin{lemma}\label{lem:pos-asso} Let $X_1, X_2, \dots$ be a sequence of independent random variables, then for any nonnegative functions $\{g_{j}(X_i, X_2, \dots)\}_{j=1}^{\infty}$ that each $g_{j}$ is either nondecreasing or nonincreasing in each $X_i$, we have
		\[
		\EE{\prod_{j=1}^{\infty}g_j(X_1, X_2, \dots)} \geq \prod_{j=1}^{\infty}\EE{g_j(X_1, X_2, \dots)}.
		\] 
		\end{lemma}
		Recall the definitions in and right after \eqnref{indicators}, and that $V$ is the number of false discoveries. 
		As for the proof for FWER control of Sidak variants, since the hypotheses are independent with each other, therefore the probability of no false discovery among all infinite decisions is 
		\begin{equation}
		    \PP{V=0} = \EE{\prod_{i \in \nulls} \one{P_i > \alpha_i} } \stackrel{(i)}{\geq} \prod_{i \in \nulls} \EE{\one{P_i > \alpha_i} } \stackrel{(ii)}{=} \prod_{i \in \nulls} \EE{\EEst{\one{P_i > \alpha_i}}{\F^{i-1}}},
		\end{equation}
		where (i) uses \lemref{pos-asso}, and (ii) uses the law of iterated expectation. Then, for Adaptive-Sidak, we have
		\begin{align*}
            \prod_{i \in \nulls} \EE{\EEst{\one{P_i > \alpha_i}}{\F^{i-1}}} 
			&\geq \prod_{i \in \nulls} \EE{ (1 - \alpha_i)} =  \prod_{i \in \nulls} \EE{(1-\alpha)^{\beta_i}}\stackrel{(ii)}{\geq} \prod_{i \in \nulls} (1-\alpha)^{\EE{\beta_i}}\\&
			\stackrel{(iii)}{\geq} \prod_{i \in \nulls} (1-\alpha)^{\EE{\beta_i\EEst{\frac{\one{P_i \geq \lambda_i}}{1-\lambda_i}}{\F^{i-1}}}}\stackrel{(iv)}{=} (1-\alpha)^{\EE{\sum_{i \in \nulls \setminus \C} \beta_i/(1-\lambda_i)}} \stackrel{(v)}{\geq} 1-\alpha,
			\end{align*}
			where (ii) is obtained using Jensen inequality, while (iii) is obatained using predictability of $\lambda_i$ and the validness of null $p$-values; (iv) are obtained using the predictability of $\lambda_i$ and $\beta_i$, and the law of iterated expectation; and (v) is true since $\sum_{i\in \nulls \setminus \C} \beta_i/(1-\lambda_i) \leq 1$ and that $1-\alpha < 1$ by construction. Hence, the probability of at least one false discovery (i.e. FWER) is at most $\alpha$.
			
			Noting that ADDIS-Sidak reduces to Discard-Sidak when setting $\lambda_i \equiv 0$ for all $i$, therefore we omit the proof of FWER control for Discard-Sidak and only present the proof for the more general algorithm ADDIS-Sidak in the following. When we additionally have that the null $p$-values are uniformly conservative as defined in \eqnref{conserve-def}, then the probability of no false discovery among all infinite decisions for ADDIS-Sidak given tested set $\S$ is:
			\begin{align*}
				\PPst{V=0}{\S} & = \EEst{\prod_{i \in \nulls \cap \S} \one{P_i > \alpha_i}}{\S} \stackrel{(i)}{\geq} \prod_{i \in \nulls \cap \S}\EEst{ \one{P_i > \alpha_i}}{\S} = \prod_{i \in \nulls \cap \S}\EEst{\EEst{ \one{P_i > \alpha_i}}{\F^{i-1},\S}}{\S} \\ 
				& \stackrel{(ii)}{\geq} \EEst{\prod_{i \in \nulls \cap \S} (1 - \alpha_i/\tau_i)}{\S} = \EEst{\prod_{i \in \nulls \cap \S} (1-\alpha)^{\beta_i}}{\S}\\
				& = \EEst{(1-\alpha)^{\sum_{i \in \nulls \cap \S} \beta_i}}{\S} \stackrel{(iii)}{\geq} (1-\alpha)^{\EEst{\sum_{i \in \nulls \cap \S} \beta_i}{\S}}\\
				& \stackrel{(iv)}{\geq} (1-\alpha)^{\EEst{\sum_{i \in \nulls \cap \S} \beta_i\EEst{\one{P_i\geq \lambda_i}/(1-\lambda_i)}{\F^{i-1},\S}}{\S}}\\
				& = (1-\alpha)^{\EEst{\sum_{i \in \nulls \cap \S \setminus{\C}} \beta_i/(1-\lambda_i)}{\S}}
			\stackrel{(v)}{\geq} (1-\alpha)
			\end{align*}
			where (i) is obtained using \lemref{pos-asso}, and (ii) is obtained using the uniformly conservative property of nulls, and the fact that the $p$-values are independent with each other, (iii) is obtained using Jesen inequality, and (iv) is obtained using the predictability of $\lambda_i$, $\tau_i$ and $\beta_i$, and the law of iterated expectation; finally (v) is true since $\sum_{i\in \nulls \cap \S\setminus{\C}} \beta_i \leq 1$ and that $1-\alpha < 1$. Hence, the probability of at least one false discovery (i.e. FWER) is at most $\alpha$. Therefore, we prove the FWER control argument for ADDIS-Sidak.		
			
			In the end, as for the proof for FWER control of Discard-Fallback, let $\nulls\cap\S = \{j_1,j_2,\dots\}$ be the possibly infinite sequence of null indices. Then, given $\S$, the event of making a false discovery is given by
			\[
			\{P_{j_1} \leq \alpha_{j_1}\} \bigcup \{P_{j_1} > \alpha_{j_1}, P_{j_2} \leq \alpha_{j_2}\} \bigcup \{P_{j_1} > \alpha_{j_1},\dots,P_{j_k} > \alpha_{j_k},P_{j_{k+1}} \leq \alpha_{j_{k+1}}\} \bigcup \dots
			\]
			Using a union bound, we have:
			\begin{align*}
				\PPst{V\geq 1}{\S} &\leq \sum_{k \geq 0} \PPst{P_{j_1} > \alpha_{j_1},\dots, P_{j_k} > \alpha_{j_k}, P_{j_{k+1}} \leq \alpha_{j_{k+1}}}{\S}\\
				&\leq \sum_{k \geq 0} \PPst{P_{j_{k+1}} \leq \alpha_{j_{k+1}}}{P_{j_1} > \alpha_{j_1},\dots, P_{j_k} > \alpha_{j_k},\S}\\
				& \stackrel{(i)}{=} \sum_{k \geq 0} \PPst{P_{j_{k+1}} \leq \alpha_{j_{k+1}}}{\S} \leq \sum_{k \geq 0} \alpha_{j_{k+1}}/\tau\\	&\stackrel{(ii)}{\leq}  \sum_{k \geq 0}\sum_{\sum_{u\leq j_k+1}S_u \leq v \leq \sum_{u\leq j_{k+1}}S_u} \gamma_v \stackrel{(iii)}{\leq} \alpha.
			\end{align*}
			where equality (i) is true since the null $p$-values are independent with each other, and inequality (ii) is true because each term in the sum is the largest possible value that $\alpha_{j_{k+1}}$ can take when $R_{j_k}=0$, and independence between the nulls means that the relevant conditional distribution of $P_{j_{k+1}}$ is still stochastically larger than uniform. Then, inequality (iii) follows because each $\gamma_v$ appears at most once in the sum. Therefore, we have:
			\[
			\textnormal{FWER} = \PP{V\geq 1} \leq \alpha,
			\]
			which complete the proof for FWER control of Discard-Fallback.
		\end{proof}
		
		\section{Closure principle for online multiple testing}\label{app:closedapp}
		We present a variant of the closure principle for online multiple testing. Given any online global null test, our closure principle constructs an algorithm that controls the FWER at level $\alpha$. Let $T_\alpha$ be a valid online global null test: (a) by online, we mean that it observes $p$-values one at a time, and at each step it either decides to stop and reject (output 1), or to observe the next $p$-value (if it never stops, that is considered to be an output of 0), and (b) by valid, we mean that if all hypotheses in the infinite sequence are truly null, then the probability that $T_\alpha$ outputs 1 at some point and stops is at most $\alpha$. The closed online testing method rejects $H_i$ if and only if $T_\alpha$ rejects every (potentially infinite) subsequence that contains $H_i$. 
   
       This method controls the FWER at level $\alpha$, as seen by the following simple argument. For there to be at least one false discovery, say $H_b$, every single subsequence involving $H_b$ must have been rejected. Specifically, the sequence involving only the nulls $\nulls$ must have been rejected by $T_\alpha$. However, we know that $\PP{T_\alpha(\nulls)=1} \leq \alpha$ by definition of validity of $T_\alpha$. Hence $\PP{\text{some null is rejected}} \leq \PP{T_\alpha(\nulls)=1} \leq \alpha$.

		\end{appendices}

\end{document}